\documentclass[reqno]{amsart}
\usepackage{amssymb,epsfig,graphics,graphicx,bbm,hyperref,color}
\usepackage{pgf,pgfarrows,pgfnodes,pgfautomata,pgfheaps,pgfshade}

\usepackage{multicol}
\usepackage{enumitem}
\usepackage{tikz-cd}
\usetikzlibrary{matrix,arrows,decorations.pathmorphing,positioning}
\usepackage{nicefrac}
\usepackage{todonotes}

\definecolor{gray}{rgb}{0.93,0.93,0.93}
\definecolor{light-gold}{rgb}{0.96,0.8,0}
\definecolor{light-red}{rgb}{1,0.4,0.4}
\definecolor{light-green}{rgb}{0.5,1,0.5}
\definecolor{light-blue}{rgb}{0.4,0.4,1}

\definecolor{gold}{rgb}{0.7,0.55,0}

\def\be{\begin{equation}}
\def\ee{\end{equation}}
\def\bm{\begin{multline}}
\def\bfig{\begin{figure}[htb]}
\def\efig{\end{figure}}
\newcommand{\dd}{{\rm d}}
\newcommand{\e}[1]{\,{\rm e}^{#1}\,}
\newcommand{\ii}{{\rm i}}

\def\Tr{{\operatorname{Tr\,}}}

\numberwithin{equation}{section}
\newtheorem{theorem}{Theorem}[section]
\newtheorem{proposition}[theorem]{Proposition}
\newtheorem{lemma}[theorem]{Lemma}

\newtheorem{thm}[theorem]{Theorem}
\newtheorem{pro}[theorem]{Proposition}
\newtheorem{lem}[theorem]{Lemma}
\newtheorem{rem}[theorem]{Remark}
\newtheorem{defi}[theorem]{Definition}
\newtheorem{corollary}[theorem]{Corollary}

\newcommand{\eps}{{\varepsilon}}
\newcommand{\bbC}{{\mathbb C}}\newcommand{\C}{{\mathbb C}}

\newcommand{\bbN}{{\mathbb N}}

\newcommand{\bbR}{{\mathbb R}}\newcommand{\R}{{\mathbb R }}

\newcommand{\frC}{{\mathfrak C}}
\newcommand{\frL}{{\mathfrak L}}

\newcommand{\cH}{\mathcal{H}}
\newcommand{\caH}{\mathcal{H}}
\newcommand{\caL}{\mathcal{L}}

\newcommand{\bss}{{\boldsymbol s}}

\def\bbone{{\mathchoice {\rm 1\mskip-4mu l} {\rm 1\mskip-4mu l} {\rm 1\mskip-4.5mu l} {\rm 1\mskip-5mu l}}}

\newcommand{\ketbra}[1]{\vert #1\rangle\langle #1\vert}
\newcommand{\ket}[1]{\vert #1 \rangle}
\newcommand{\eq}[1]{(\ref{#1})}

\makeatletter
  \def\tagform@#1{\maketag@@@{\scriptsize{(#1)}\@@italiccorr}}
\makeatother

\renewcommand{\eqref}[1]{(\ref{#1})}

\newcommand{\G}{\Gamma}

\newcommand{\ol}{\overline}
\newcommand{\oo}{\infty}
\newcommand{\om}{\omega}
\newcommand{\deq}{\mathrel{\mathop:}=}
\newcommand{\up}{\uparrow}
\newcommand{\dn}{\downarrow}
\newcommand{\g}{\gamma}
\newcommand{\se}{\subseteq}
\newcommand{\fX}{\mathfrak{X}}
\newcommand{\fA}{\mathfrak{A}}
\newcommand{\fC}{\mathfrak{C}}
\newcommand{\s}{\sigma}
\renewcommand{\S}{\Sigma}
\newcommand{\cT}{\mathcal{T}}
\newcommand{\one}{\hbox{\rm 1\kern-.27em I}}

\newcommand{\cW}{\mathcal{W}}
\newcommand{\cL}{\mathcal{L}}
\newcommand{\N}{{\mathbb N}}
\newcommand{\cO}{\mathcal{O}}

\makeatletter
\newcommand\xleftrightarrow[2][]{%
  \ext@arrow 9999{\longleftrightarrowfill@}{#1}{#2}}
\newcommand\longleftrightarrowfill@{%
  \arrowfill@\leftarrow\relbar\rightarrow}
\makeatother

\newcommand{\cross}{\mathchoice
{\vcenter{\hbox{\includegraphics[width=.8em]{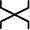}}}}
{\vcenter{\hbox{\includegraphics[width=.8em]{cross.pdf}}}}
{\vcenter{\hbox{\includegraphics[width=.6em]{cross.pdf}}}} 
{\vcenter{\hbox{\includegraphics[width=.5em]{cross.pdf}}}} 
}
\newcommand{\dbar}{\mathchoice
{\vcenter{\hbox{\includegraphics[width=.8em]{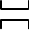}}}}
{\vcenter{\hbox{\includegraphics[width=.8em]{dbar.pdf}}}}
{\vcenter{\hbox{\includegraphics[width=.6em]{dbar.pdf}}}} 
{\vcenter{\hbox{\includegraphics[width=.5em]{dbar.pdf}}}} 
}
\newcommand{\dbartop}{\mathchoice
{\vcenter{\hbox{\includegraphics[width=.8em]{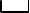}}}}
{\vcenter{\hbox{\includegraphics[width=.8em]{dbar-top.pdf}}}}
{\vcenter{\hbox{\includegraphics[width=.6em]{dbar-top.pdf}}}} 
{\vcenter{\hbox{\includegraphics[width=.5em]{dbar-top.pdf}}}} 
}
\newcommand{\dbarbot}{\mathchoice
{\vcenter{\hbox{\includegraphics[width=.8em]{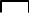}}}}
{\vcenter{\hbox{\includegraphics[width=.8em]{dbar-bot.pdf}}}}
{\vcenter{\hbox{\includegraphics[width=.6em]{dbar-bot.pdf}}}} 
{\vcenter{\hbox{\includegraphics[width=.5em]{dbar-bot.pdf}}}} 
}

\begin{document}


\title{Dimerization in quantum spin chains with $O(n)$ symmetry}

\author[J.E. Bj\"ornberg]{Jakob E. Bj\"ornberg}
\address{Department of Mathematics,
Chalmers University of Technology and the University of Gothenburg,
Sweden}
\email{jakob.bjornberg@gu.se}

\author[P. M\"uhlbacher]{Peter M\"uhlbacher}
\address{Department of Mathematics, University of Warwick,
Coventry, CV4 7AL, United Kingdom}
\email{peter@muehlbacher.me}

\author[B. Nachtergaele]{Bruno Nachtergaele}
\address{Department of Mathematics and Center for Quantum Mathematics and Physics\\
University of California, Davis\\
Davis, CA 95616, USA}
\email{bxn@math.ucdavis.edu}

\author[D. Ueltschi]{Daniel Ueltschi}
\address{Department of Mathematics, University of Warwick,
Coventry, CV4 7AL, United Kingdom}
\email{daniel@ueltschi.org}



\begin{abstract}
We consider quantum spins with $S\geq1$, and two-body interactions with $O(2S+1)$ symmetry. We discuss the ground state phase diagram of the one-dimensional system. We give a rigorous proof of dimerization for an open region of the phase diagram, for $S$ sufficiently large. We also prove the existence of a gap for excitations.
\end{abstract}

\thanks{\copyright{} 2021 by the authors. This paper may be reproduced, in its
entirety, for non-commercial purposes.}

\maketitle

\tableofcontents

\section{Introduction}

Over the course of almost a century of studying quantum spin chains, physicists and mathematicians have uncovered a wide variety of interesting 
physical phenomena and in the process invented an impressive arsenal of new mathematical techniques and structures. Nevertheless, our understanding
of these simplest of quantum many-body systems is still far from complete. For many models of interest we have only partial information about the ground state phase diagram, the nature of the phase transitions, and the spectrum of excitations. We consider here a family of spin systems with two-body interactions, where interactions are translation invariant and $O(2S+1)$ invariant. We investigate the ground state phase diagram, looking for ground states that possess less symmetry than the interactions. Our main result is a rigorous proof of dimerization (where translation invariance is broken) in a region of the phase diagram with $S$ large enough (Theorem \ref{thm main}). We also prove exponential clustering (Theorem \ref{thm exp decay}) and the existence of a gap (Theorem \ref{thm gap}).

The family of models is introduced in Section \ref{sec model}; the phase diagram for general $S\geq1$ is described in Section \ref{sec phd general}; the case $S=1$ has received a lot of attention and we discuss it explicitly in Section \ref{sec phd 3}; our result about dimerization is stated in Section \ref{sec result}.

The $O(n)$ models have a graphical representation which we describe in Section \ref{sec loops}. We use it to define a ``contour model" in Section \ref{sec contours} where contours are shown to have small weights. This allows to use the method of cluster expansion and prove dimerization in Section \ref{sec dimerisation}.

\subsection{A family of quantum spin chains with $O(n)$-invariant interactions}
\label{sec model}

We consider a family of quantum spin chains consisting of $2\ell$ spins of magnitude $S$ defined by a nearest-neighbor Hamiltonian $H_\ell$ acting
on the Hilbert space $\cH_\ell= (\C^{n})^{\otimes 2\ell}$, with $n=2S+1\geq 2$, of the form
\be\label{ham}
H_\ell = \sum_{x=-\ell +1}^{\ell-1} h_{x,x+1},
\ee
where  $h_{x,x+1}$ denotes a copy of $h=h^*\in M_{n}(\C)\otimes M_{n}(\C)$ acting on the nearest neighbor pair at sites $x$ and $x+1$.

We are interested in the family of interactions 
\be\label{hab}
h= uT + vQ, \quad u,v\in\R,
\ee
where $T$ is the transposition operator defined by $T(\phi\otimes \varphi) = \varphi\otimes \phi$, for $\phi,\varphi\in \C^{n}$, and $Q$ is the orthogonal projection 
onto the one-dimensional subspace of  $\C^{n}\otimes  \C^{n}$ spanned by a vector of the form
\be\label{psi}
\psi = \frac{1}{\sqrt{n}} \sum_{\alpha = 1}^n e_\alpha\otimes e_\alpha,
\ee
for some orthornormal basis $\{ e_\alpha| \alpha = 1, \ldots ,n\}$ of $\C^n$.
%

The spectrum of $h$ is easy to find. $T$ has the eigenvalues $1$ and $-1$, corresponding to the symmetric and antisymmetric subspaces
of $\C^n\otimes  \C^n$, whose dimensions are $n(n+1)/2$ and $n(n-1)/2$, respectively. Since $\psi$ is symmetric, the eigenvalues of $h$ are
$u+v, u, -u$.

Let $R$ 
be a linear transformation represented 
by an orthogonal matrix in the basis $\{e_\alpha\}$,
meaning $\langle e_\alpha, RR^{\rm T} e_\beta\rangle = \delta_{\alpha\beta}$. This amounts to defining a specific representation
of $O(n)$ on the system under consideration. It is then straightforward to check $(R\otimes R) \psi  = \psi$. It follows that $R\otimes R$
commutes with $Q=\ketbra{\psi}$. Since $T$ also commutes with $R\otimes R$, the Hamiltonians with interaction $h$ given in \eq{hab}
have a local $O(n)$ symmetry. This family of models is in fact, up to a trivial additive constant, the most general translation-invariant 
nearest neighbor Hamiltonian for spins of dimension $n$ and with a translation-invariant local $O(n)$ symmetry.


To make contact with previous results in the literature, it is useful to note a couple of equivalent forms of the spin chains we consider.
First, for integer values of $S$, that is odd dimensions $n$, consider the orthonormal basis $\{e_\alpha\}$, relabeled by $\alpha = -S,\ldots, S$, 
and related to the standard eigenbasis of the third spin matrix $S^{(3)}$, satisfying $S^{(3)} \ket{\alpha} = \alpha\ket{\alpha}$, as follows:
for $\alpha=0$ take $e_0 = \ii^S \ket{0}$, and for $\alpha>0$ define
\be\label{takagi_basis}
e_\alpha = \frac{\ii^{S-\alpha}}{\sqrt2} \bigl( \ket{\alpha} + \ket{-\alpha} \bigr), \quad e_{-\alpha} = \frac{\ii^{S-\alpha+1}}{\sqrt2} \bigl( \ket{\alpha} - \ket{-\alpha} \bigr).
\ee
Then, we have
\be
\label{phi}
\psi = \phi:=
\frac{1}{\sqrt{n}} \sum_{\alpha = -S}^S  (-1)^{S-\alpha} |\alpha,-\alpha\rangle,
\ee
which is the $SU(2)$ singlet vector in the standard spin basis. The transposition operator $T$ is of course not affected by any translation-invariant local basis change. Therefore, for odd $n$, and with a simple change of basis, the family of interactions
\eq{hab} is seen to be equivalent to
\be\label{tildehab}
\tilde{h}= uT + vP, \quad u,v\in\R,
\ee
where $P$ is the orthogonal projection onto the singlet state
$\phi$.

The case of even $n$ is different. Interactions $h$ and $\tilde h$ are not unitarily equivalent. But the model with interaction $\tilde h$ is nonetheless interesting and we discuss it in Appendix \ref{app P}. We also prove dimerization and a gap in this case, see Theorem \ref{thm P} and Theorem \ref{thm gap}.

For $n\geq 2$, $u=0$, and $v=-1$, this is the much studied $-P^{(0)}$ spin chain
\cite{barber:1989,affleck:1990,klumper:1990,AN,nepomechie:2016,NU,ADW}.

\subsection{Ground state phase diagram for general $n\geq3$}
\label{sec phd general}

We start with the phase diagram for arbitrary $n\geq3$ and discuss the special case $n=3$ in Section \ref{sec phd 3}.
The ground state phase diagram of the spin chain with nearest-neighbor interactions $h_{x,x+1} = uT_{x,x+1} + vQ_{x,x+1}$ is depicted in Fig.\ \ref{fig phd general}. It can be broadly divided into four domains.

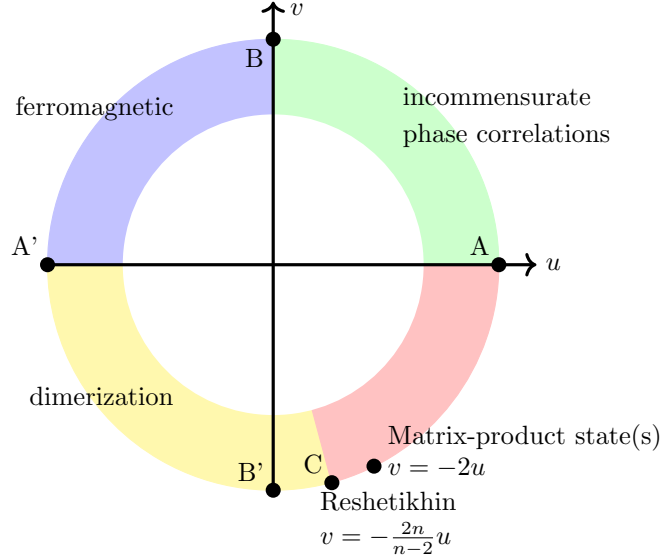
\begin{figure}
\begin{tikzpicture}
\filldraw[fill=light-blue!40,draw=light-blue!40] (0,0) -- (0,3) arc (90:180:3) -- cycle;
\filldraw[fill=yellow!40,draw=yellow!40] (0,0) -- (-3,0) arc (180:285:3) -- cycle;
\filldraw[fill=light-green!40,draw=light-green!40] (0,0) -- (3,0) arc (0:90:3) -- cycle;
\filldraw[fill=light-red!40,draw=light-red!40] (0,0) -- (3,0) arc (0:-75:3) -- cycle;

\fill[fill=white] (0,0) circle (2cm);

\draw[very thick,->] (0,-3) -- (0,3.5);
\draw[very thick,->] (-3,0) -- (3.5,0);

\draw[right] (3.5,0) node(1){$u$};
\draw[above right] (0.1,3.2) node(1){$v$};
\fill (3,0) circle (0.1cm);
\draw[above left]  (3,0) node(){\rm A};  
\fill (-3,0) circle (0.1cm);
\draw[above left]  (-3,0) node(){\rm A'};  
\fill (0,3) circle (0.1cm);
\draw[below left]  (0,3) node(){\rm B};  
\fill (0,-3) circle (0.1cm);
\draw[above left]  (0,-3) node(){\rm B'};  
\fill (0.78,-2.9) circle (0.1cm);
\draw[above left] (0.78,-2.9) node(){\rm C};  
\draw[below right] (0.5,-2.9) node(){Reshetikhin};  
\draw[below right] (0.5,-3.3) node(){$v = -\frac{2n}{n-2} u$};  
\draw[above left] (-1.2,1.8) node(){ferromagnetic};  
\draw[below left] (-1.2,-1.5) node(){dimerization};  
\draw[above right] (1.6,2) node(){incommensurate};  
\draw[below right] (1.6,2) node(){phase correlations};  
\fill (1.34,-2.68) circle (0.1cm);
\draw[above right] (1.4,-2.6) node(){Matrix-product state(s)};  
\draw[above right] (1.4,-2.95) node(){$v= -2u$};  

\end{tikzpicture}
\caption{Ground state phase diagram for the chain with nearest-neighbor interactions $uT + vQ$ for $n\geq3$. Our main result, Theorem \ref{thm main}, is a proof of dimerization in an open region around the point B'.}
\label{fig phd general}
\end{figure}

The domain formed by the quadrant $u\leq0, v\geq0$ (blue region in Fig.\ \ref{fig phd general}) is ferromagnetic. There are many ground states and they minimize $h_{x,x+1}$ for all $x$; that is, they are frustration-free. The ground state energy per bond is equal to $u$. Indeed, let $\varphi = \sum_\alpha c_\alpha e_\alpha$ with $\sum_\alpha |c_\alpha|^2 = 1$. It is clear that $|\varphi \otimes \varphi \rangle$ is eigenstate of $T$ with eigenvalue 1; further, we have
\be 
\langle \varphi \otimes \varphi | Q | \varphi \otimes \varphi \rangle 
= \frac1n \Bigl| \sum_\alpha c_\alpha^2 \Bigr|^2.
\label{condition for zero ev}
\ee
The latter is zero when $\sum_\alpha c_\alpha^2 = 0$. Since $Q$ is a projector, such a state is eigenstate with eigenvalue 0. Notice that the state $R\varphi$ also satisfies this condition, for all orthogonal transformation $R$.

The product state $\otimes_{x=-\ell+1}^\ell \varphi$ is then a ground state of $h_{x,x+1}$ with eigenvalue $u$, for all $x$. In addition to these product states, we can obviously take linear combinations.

The next domain is the arc-circle between $(u,v) = (-1,0)$ and the
``Reshetikhin point" with $v = -\frac{2n}{n-2}u$ (yellow region in
Fig.\ \ref{fig phd general}), 
which features dimerization. In order to see that 
dimerization is plausible as soon as $v<0$, let
\be
\varphi_{x,x+1} = \sqrt{1-\eps^2} \, |S,S\rangle + \frac\eps{\sqrt{n-1}} \sum_{\alpha=-S}^{S-1} |\alpha,\alpha\rangle.
\label{partially dimerised}
\ee
and consider the (partially) dimerized state
$\varphi_{-\ell+1,-\ell+2} \otimes \varphi_{-\ell+3,-\ell+4} \otimes \dots$.
For $\eps=0$, this is a product state, but for $\eps\neq0$ it is not. 
Roughly half the edges, namely the edges $x,x+1$  with 
$x=-\ell+1,-\ell+3,\dotsc$, 
 are dimerized and their energy is
\be
\langle \varphi_{x,x+1} | uT_{x,x+1} + vQ_{x,x+1} | \varphi_{x,x+1} \rangle = u + \frac vn (1 + 2\sqrt{n-1} \eps) + O(\eps^2).
\ee
The non-dimerized edges contribute
\be
\langle \varphi_{x-1,x} \otimes \varphi_{x+1,x+2} | uT_{x,x+1} + vQ_{x,x+1} | \varphi_{x-1,x} \otimes \varphi_{x+1,x+2} \rangle = u+ \frac vn + O(\eps^2).
\ee
The average energy per bond of the state $\varphi_{x,x+1}$ is then 
$u+ \frac vn + v \frac{\sqrt{n-1}}n \eps$, up to $O(\eps^2)$
corrections. When $v<0$ the optimal product states have energy $u +
\frac vn$ (using \eqref{condition for zero ev} with $\sum_\alpha
c^2_\alpha = 1$), so the partially dimerized state \eqref{partially
  dimerised} has lower energy when $\eps$ is positive and small.  

Our main result is that dimerization does occur in an open domain around the point B', provided $n$ is sufficiently large, see Theorem \ref{thm main}. This extends the results of \cite{NU,ADW}, valid at the point B'.

Then comes the domain formed by the arc-circle between the Reshetikhin point $v = -\frac{2n}{n-2}u$ and $(u,v)=(1,0)$ (red region in Fig.\ \ref{fig phd general}). For $n$ odd a unique translation-invariant ground state is expected. 

This domain contains several interesting special cases. The direction
$(u=1,v=0)$ is the the $SU(n)$ generalization of the spin-1/2
Bethe-ansatz solvable Heisenberg model studied by Sutherland and
others \cite{Sutherland}. The direction $v = -\frac{2n}{n-2} u$ was solved by
Reshetikhin \cite{reshetikhin:1983} (this generalizes the
Takhtajan--Babujian model for $n=3$). These models are gapless. The
direction $v=-2u$ is a frustration free point and the ground states
are given matrix-product states. For odd $n$, these are generalizations
of the AKLT model. The ground state for the
infinite chain is unique and is in the Haldane phase. For even $n$,
there are two matrix-product ground states that break the translation
invariance of the chain down to period 2 \cite{Tu}.

The final domain is the quadrant $u,v>0$. The ground states are
expected to have slow decaying correlations with incommensurate phase
correlations. That is, spin-spin correlations between sites 0 and $x$
are expected to be of the form $|x|^{-r} \cos( \omega |x|)$ for $|x|$
large, and where $r,\omega$ depend on the parameters $u,v$ \cite{Itoi}. 

It is perhaps worth mentioning that the phase diagram for spatial
dimensions other than 1 is quite different. Dimerization is not
expected. Instead, the system displays various forms of magnetic
long-range orders (ferromagnetic,  spin nematic, N\'eel, \dots). See
\cite{Uel} for results about magnetic ordering for all $n\geq2$ and
for parameters that correspond to the dimerized phase here.

\subsection{The $S=1$ model ($n=3$)}
\label{sec phd 3}

For $n=3$, the family of models is equivalent to the familiar spin-1 chain with bilinear and biquadratic interactions. The latter is most often parametrized by an angle $\phi$ as follows:
\be\label{bilin-biquad}
 \cos \phi \,\vec{S}_x\cdot \vec{S}_{x+1} + \sin \phi\, (\vec{S}_x\cdot \vec{S}_{x+1})^2 = 3 (\sin\phi - \cos\phi)P + \cos\phi\, T +\sin\phi\, I.
\ee
We can apply the change of basis that is the inverse of Eq.\ \eqref{takagi_basis}, namely
\be
|0\rangle = -\ii \, e_0, \quad |1\rangle = \tfrac1{\sqrt2} (e_1 - \ii \, e_{-1}), \quad |-1\rangle = \tfrac1{\sqrt2} (e_1 + \ii \, e_{-1}).
\ee
Then the interaction is given by \eqref{bilin-biquad} but with the operator $Q$ instead of $P$.

\begin{figure}
\begin{tikzpicture}
\filldraw[fill=light-blue!40,draw=light-blue!40] (0,0) -- (0,3) arc (90:225:3) -- cycle;
\filldraw[fill=yellow!40,draw=yellow!40] (0,0) -- (0,-3) arc (270:225:3) -- cycle;
\filldraw[fill=yellow!40,draw=yellow!40] (0,0) -- (0,-3) arc (270:315:3) -- cycle;
\filldraw[fill=light-green!40,draw=light-green!40] (0,0) -- (0,3) arc (90:45:3) -- cycle;
\filldraw[fill=light-red!40,draw=light-red!40] (0,0) -- (3,0) arc (0:45:3) -- cycle;
\filldraw[fill=light-red!40,draw=light-red!40] (0,0) -- (3,0) arc (0:-45:3) -- cycle;

\fill[fill=white] (0,0) circle (2cm);

\draw[very thick,->] (0,-3) -- (0,3.5);
\draw[very thick,->] (-3,0) -- (3.5,0);

\draw[right] (3.5,0.02) node(1){$\cos\phi$};
\draw[above] (0,3.4) node(1){$\sin\phi$};
\fill (2.12,2.12) circle (0.1cm);
\draw[left] (2.1,2.14) node(){\rm A};  
\draw[right] (2.2,2.14) node(){\rm Sutherland};  
\fill (-2.12,-2.12) circle (0.1cm);
\draw[below left]  (-2.12,-2.12) node(){\rm A'};  
\fill (0,3) circle (0.1cm);
\draw[below left]  (0,3) node(){\rm B};  
\fill (0,-3) circle (0.1cm);
\draw[above left]  (0,-3) node(){\rm B'};  
\fill (2.12,-2.12) circle (0.1cm);
\draw[above left] (2.12,-2.12) node(){\rm C};  
\draw[below right] (2.12,-2.12) node(){Takhtajan-Babujian};
\fill (2.85,0.95) circle (0.1cm);
\draw[right]  (2.9,0.95) node(){\rm AKLT ($\tan\phi = \frac13$)};  
\draw[above left] (-1.2,1) node(){ferromagnetic};  
\draw[right] (-1,-2.3) node(){dimerization};  
\draw[above right] (0.4,2.9) node(){incommensurate};  
\draw[below right] (0.4,2.9) node(){phase correlations};  

\end{tikzpicture}
\caption{Ground state phase diagram for the $S=1$ chain with nearest-neighbor interactions $\cos\phi \vec S_{x} \cdot \vec S_{x+1} + \sin\phi (\vec S_x \cdot \vec S_{x+1})^2$. The domains and the points are the same as those in Fig.\ \ref{fig phd general}.}
\label{fig phd 3}
\end{figure}
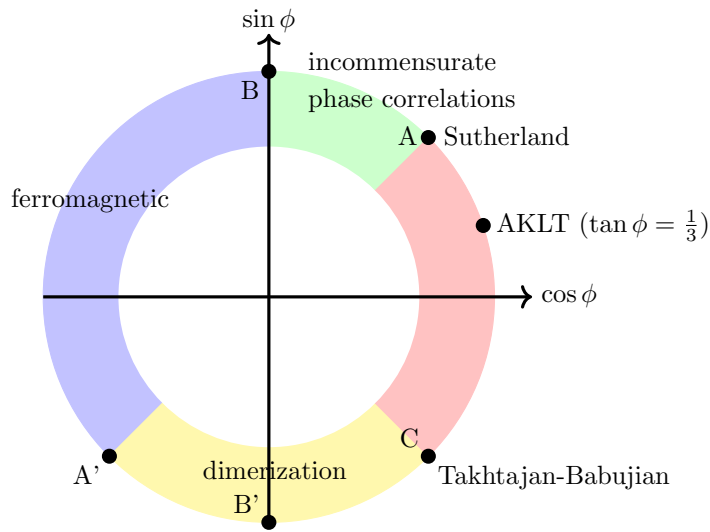

The ground state phase diagram with parameter $\phi$ is depicted in Fig.\ \ref{fig phd 3}. The domains and the points are the same as in Fig.\ \ref{fig phd general}. The ferromagnetic domain corresponds to $\phi \in (\frac\pi2, \frac{5\pi}4)$,
and the model is frustration-free in this range. Among the ground states, there is a family of product states that shows
that the $O(3)$ symmetry of the Hamiltonian is spontaneously broken. As a consequence, the Goldstone Theorem 
\cite{landau:1981} implies that there are gapless excitations above the ground state in this region.
The dimerization domain is $\phi \in (\frac{5\pi}4, \frac{7\pi}4)$. The next domain is $\phi \in (-\frac\pi4, \frac\pi4)$ with unique, translation-invariant ground states. Finally, the domain $\phi \in (\frac\pi4, \frac\pi2)$ is expected to display states with slow decay of correlations, with incommensurate phase.

There are several points where exact and/or rigorous information is
available: (i) $\phi\in [0,\pi/2]$ with $\tan \phi = 1/3$, it is the
spin-1 AKLT chain \cite{AKLT} with interaction $\tilde{h}$ given by
the orthogonal projection on the spin-2 states. In the thermodynamic
limit, it has a unique ground state of Matrix Product form with a
non-vanishing spectral gap and exact exponential decay of
correlations; (ii) the two points with $\tan \phi =1$, A and A' in
Fig. \ref{fig phd 3}, have $SU(3)$ symmetry and are often referred to
as the Sutherland model \cite{Sutherland}. An exact solution for the ground
state at $\phi = -3\pi/4$ is gapless and highly degenerate, while for
$\phi=\pi/4$ is believed to be a unique critical state with gapless
excitations; (iii) the point $\phi = -\pi/4$ is the Bethe-ansatz
solvable Takhtajan--Babujian model \cite{Takhtajan,Babujian}, which is also gapless;
(iv) the point $\phi=-\pi/2$, is the $-P^{(0)}$ spin-1 chain, already
mentioned above. Aizenman, Duminil-Copin, and Warzel proved that it
has two dimerized (2-periodic) ground states with exponential decay of
correlations \cite{ADW}; all evidence indicates that these states are
gapped. 

Let us briefly comment on higher spatial dimensions. Dimerization is
not expected. Various rigorous results about magnetic long-range order have
been established: for $\phi = 0$ \cite{DLS}; for $\phi \gtrsim
\frac{5\pi}4$ \cite{TTI,Uel}; and for $\phi \lesssim 0$
\cite{Lees}. Recently, the model on the complete graph has been
studied by Ryan using methods based on the Brauer algebra \cite{Ryan}, 
which plays a role in the representation theory of the orthogonal groups analogous
to that of the symmetric group for the general linear groups.

\subsection{Our result about dimerization}  
\label{sec result}

Let us introduce the operators $L^{\alpha,\alpha'}$, $1 \leq \alpha <
\alpha' \leq n$, that are generators of the Lie algebra $\mathfrak{o}(n)$: 
\be
L^{\alpha,\alpha'} = |\alpha\rangle \langle \alpha'| - |\alpha'\rangle \langle \alpha|.
\ee
And for $x \in \{-\ell+1,\dots,\ell\}$, let $L_x^{\alpha,\alpha'}$ be
the operator in $\caH_\ell$ that acts as $L^{\alpha,\alpha'}$ at the
site $x$, and as the identity elsewhere. 


\begin{theorem}
\label{thm main}
There exist constants $n_0, u_0, c >0$ (independent of $\ell$) such that for $n > n_0$ and 
$|u| < u_0$, we have that for all $1 \leq \alpha < \alpha' \leq n$,
\[
\begin{split}
&\lim_{\beta\to\infty} \Bigl[ \langle L_{0}^{\alpha,\alpha'} L_{1}^{\alpha,\alpha'}
\rangle_{\ell,\beta,u} - 
\langle L_{-1}^{\alpha,\alpha'} L_{0}^{\alpha,\alpha'} \rangle_{\ell,\beta,u} \Bigr] > c \qquad \text{for all $\ell$ odd;} \\
&\lim_{\beta\to\infty} \Bigl[ \langle L_{0}^{\alpha,\alpha'} L_{1}^{\alpha,\alpha'}
\rangle_{\ell,\beta,u} 
- \langle L_{-1}^{\alpha,\alpha'} L_{0}^{\alpha,\alpha'} \rangle_{\ell,\beta,u} \Bigr] < -c \qquad \text{for all $\ell$ even.}
\end{split}
\]
\end{theorem}

\begin{figure}[htb]\center
\includegraphics[scale=.75]{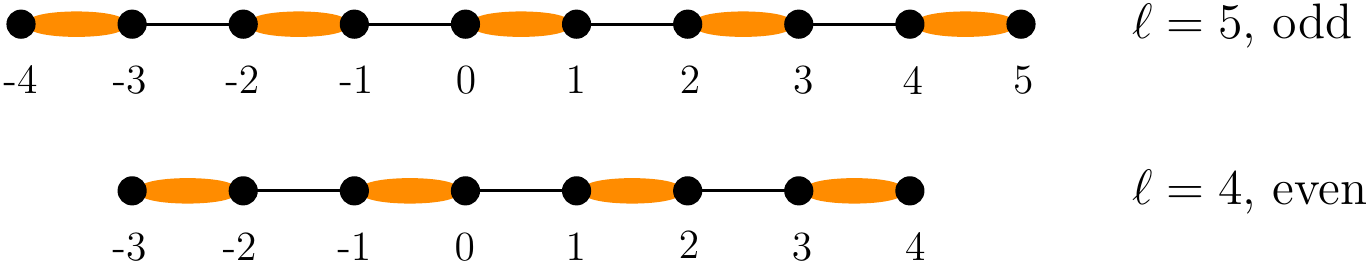}
\caption{Illustration for dimerization. Depending on whether $\ell$ is even or odd, the site $x=0$ is more entangled with its left or its right neighbor.}
\label{fig dimers}
\end{figure}

Theorem \ref{thm main} establishes the existence of at least two distinct infinite-volume ground states, close to the point B' of the phase diagram (see Fig.\ \ref{fig dimers}). Notice that the same result holds if we replace the operators $L_0^{\alpha,\alpha'} L_1^{\alpha,\alpha'}$ with spin operators $S_0^{(3)} S_1^{(3)}$, diagonal in the basis $\{e_\alpha\}$.

We expect that there are exactly two extremal ground states, precisely given by limits $\ell \to \infty$ along odd or even integers. We also expect that, if we take the chain to be $\{ -\ell, -\ell+1, \dots, \ell\}$, the corresponding infinite-volume ground state is equal to the average of the two extremal states.

The next result shows that the ground state retains the $O(n)$ symmetry of the system, that there is no magnetic long range order. This is indeed an attribute of dimerisation.

\begin{theorem}
\label{thm exp decay}
There exist constants $n_0, u_0, c_1, c_2, C >0$ (independent of $\ell$)
such that for $n > n_0$ and  $|u| < u_0$, we have
\[
\lim_{\beta\to\infty} \bigl| \langle L_{x}^{\alpha,\alpha'} \e{-t H_\ell} L_{y}^{\alpha,\alpha'} \e{t H_\ell} \rangle_{\ell,\beta,u} \bigr| \leq C \e{-c_1 |x-y| - c_2 |t|}
\]
for all $\ell\in \bbN$, all $x,y \in \{-\ell+1,\dots,\ell\}$, all $1 \leq \alpha < \alpha' \leq n$, and all $t \in \bbR$.
\end{theorem}

Dimerization has been established in \cite{NU,ADW} at the point B' in
the phase diagrams of Figs \ref{fig phd general} and \ref{fig phd
  3}. The earlier result \cite{NU} uses the loop representation of
\cite{AN} combined with a Peierls argument; it holds for $S \geq 8$
(or $n\geq17$). The second result, due to Aizenman, Duminil-Copin and
Warzel, remarkably holds for all $S \geq 1$ ($n\geq3$), i.e., for all
values of $S$ (or $n$) where dimerization is expected. It uses the
loop representation and random cluster representation of \cite{AN} as
well as recent results for the two-dimensional random cluster model
\cite{DLM,RS}. 

Away from the point B' these methods do not apply. In this article we use
the loop representation of \cite{Uel}, which combines those of
\cite{Toth,AN}, in order to get a contour model; see Theorem \ref{thm spin loop}. The loop
representation involves a probability measure for $u,v\leq0$ only; it
involves a signed measure otherwise. This is described in Section \ref{sec loops}. For large $n$, typical configurations involve many loops, that are short loops located on all the dimerized edges. We define contours to be excitations with respect to this background. It is possible to obtain a contour model with piecewise compatible contours, that is suitable for a cluster expansion (Section \ref{sec contours}). This method is robust regarding signs and it allows to intrude in the region with positive parameter $u$. Proving that the expansion converges is difficult, since the cost of excitations is entropic rather than energetic. This is done in Section \ref{sec dimerisation}. This allows to establish dimerization in the loop model, see Theorem \ref{thm loop dimer}. It is equivalent to Theorem \ref{thm main}, thus proving our main result. Theorem \ref{thm exp decay} is proved in Subsection \ref{sec exp decay}.

The interaction that is responsible for dimerization is the operator $Q_{x,x+1}$ and we prove that dimerization
is stable under perturbations of this interaction by $u T_{x,x+1}$, with $|u|$ sufficiently small. It should be possible to prove stability under more general perturbations that are not necessarily invariant under the group $O(n)$. Since the
unperturbed model is {\em not} frustration-free, this does not follow from the recent result about the stability of 
gapped phases with discrete symmetry breaking in \cite{NSY2020}, which requires the frustration-free property. For translation-invariant perturbations by $O(n)$ invariant next-nearest neighbor or further terms, the methods of this
paper should generalize in a straightforward manner.

We now discuss the case of the spin chain with Hamiltonian
\be
\label{def H tilde}
\tilde H_\ell = \sum_{x=-\ell+1}^{\ell-1} 
\bigl( uT_{x,x+1} + vP_{x,x+1} \bigr),
\ee
where $P$ is projection onto the singlet state (recall \eqref{tildehab}).
We have a similar result about dimerization. Let
$S^{(i)}$, $i=1,2,3$, be the spin operators that are the generators of the $SU(2)$ symmetry group for $\tilde H_\ell$. In the basis $|\alpha\rangle$ where $P$ is the projection onto the vector $\phi$ in \eqref{phi}, we can choose $S^{(3)}$ such that $S^{(3)} \ket{\alpha} = \alpha\ket{\alpha}$. Let
\be
\langle S_x^{(i)} S_y^{(i)} \rangle_{\ell,\beta,u}^{\tilde{}} = 
\frac1{\Tr \e{-\beta \tilde H_\ell}} \Tr S_x^{(i)} S_y^{(i)} \e{-\beta \tilde H_\ell}.
\ee

\begin{theorem}
\label{thm P}
Let $v=-1$, and $i \in \{1,2,3\}$. There exist constants $n_0, u_0, c >0$
(independent of $\ell$) such that for $n > n_0$ and  
$|u| < u_0$, we have
\[
\begin{split}
&\lim_{\beta\to\infty} \Bigl[ \langle S_0^{(i)} S_1^{(i)}
\rangle_{\ell,\beta,u}^{\tilde{}} - 
\langle S_{-1}^{(i)} S_0^{(i)} \rangle_{\ell,\beta,u}^{\tilde{}} \Bigr] > c \qquad \text{for all $\ell$ odd;} \\
&\lim_{\beta\to\infty} \Bigl[ \langle S_0^{(i)} S_1^{(i)}
\rangle_{\ell,\beta,u}^{\tilde{}} 
- \langle S_{-1}^{(i)} S_0^{(i)} \rangle_{\ell,\beta,u}^{\tilde{}} \Bigr] < -c \qquad \text{for all $\ell$ even.}
\end{split}
\]
\end{theorem}

When $n$ is odd this theorem is equivalent to Theorem \ref{thm main}, as the correlations of spin operators are the same as correlations of operators $L_{x,y}^{\alpha,\alpha'}$, up to some factors. In the case where $n$ is even, this is no longer the case and the proof needs to be adapted; the modifications are described in Appendix \ref{app P}.

\subsection{Gap for excitations}

Let $E_0^{(\ell)} < E_1^{(\ell)} < \dots$ be the eigenvalues of $H_\ell$, and $\tilde E_0^{(\ell)} < \tilde E_1^{(\ell)} < \dots$ be the eigenvalues of $\tilde H_\ell$. The gaps are defined as
\be
\begin{split}
\Delta^{(\ell)} = E_1^{(\ell)} - E_0^{(\ell)}, \\
\tilde\Delta^{(\ell)} = \tilde E_1^{(\ell)} - \tilde E_0^{(\ell)}.
\end{split}
\ee
The gaps are obviously positive but the question is whether they are so uniformly in $\ell$.

\begin{theorem}
\label{thm gap}
There exist constants $n_0, u_0, c >0$ (independent of $\ell$) such that for $n > n_0$ and 
$|u| < u_0$, we have 
\begin{itemize}
\item[(a)] The multiplicities of $E_0^{(\ell)}$ and $\tilde E_0^{(\ell)}$ are equal to 1. (That is, ground states are unique.)
\item[(b)] $\Delta^{(\ell)} \geq c$ and $\tilde \Delta^{(\ell)} \geq c$ for all $\ell$.
\end{itemize}
\end{theorem}

Recall that the chain is $\{-\ell+1, \dots, \ell\}$ and it always contains an even number of sites. Our theorem does not cover the chains with odd numbers of sites, although we expect the corresponding Hamiltonians to be gapped as well.

The spatial exponential decay proved in Theorem \ref{thm exp decay} is also a consequence of Theorem \ref{thm gap}, due to the Exponential Clustering Theorem (see the simultaneous articles \cite{HK,NS}). Our proof here is motivated by \cite{KT}. For the model $H_\ell$ it can be found in Section \ref{sec gap}. It relies on a loop and contour representation, and on cluster expansions, as for the proof of dimerization. The modifications for $\tilde H_\ell$ are discussed in the appendix.

\section{Graphical representation for $O(n)$ models}
\label{sec loops}


Consider the one-dimensional
graph consisting of the
$2\ell$ vertices $V_\ell\deq \{-\ell+1,\dots,\ell\}$ and the edges
$E_\ell\deq \big\{(x,x+1):-\ell+1\leq x\leq \ell-1\big\}$.
Fix  $\beta>0$.
To each vertex and edge of this graph we associate a 
 periodic time interval
$T_\beta=(-\beta,\beta)_\text{per}$ to obtain
a set of \emph{space-time vertices} 
$\overline V_{\ell,\beta}\deq V_\ell\times T_\beta$ as well as
a set of 
\emph{space-time edges} $\overline E_{\ell,\beta}\deq E_\ell\times T_\beta$.

By a \emph{configuration} $\omega$ we mean a finite
subset of $\overline E_{\ell,\beta}$, each point of $\omega$ receiving a
\emph{mark} $\cross$ or $\dbar$.  The points of $\omega$ will 
collectively be called \emph{links}, those marked $\cross$ being
referred to as \emph{crosses} and those marked $\dbar$ as 
\emph{double-bars}.
We write $\omega=(\omega_{\dbar},\omega_{\cross})$ and
 denote the set of all such
(link) configurations $\Omega_{\ell,\beta}$.

To every configuration $\omega\in\Omega_{\ell,\beta}$ corresponds a 
set of \emph{loops}; see Fig.\ \ref{fig:loop-sorts}
for an illustration.  A loop $l$ is a closed, injective trajectory 
\begin{align*}
	[0,L]_\text{per}&\to \overline V_\ell\\
	t&\mapsto l(t)=(v(t),T(t)),
\end{align*}
such that $x(t)$ is piecewise constant and $T'(t)\in\{\pm 1\}$. We
call $L\equiv |l|$ the \emph{length} of $l$, that is the smallest
$L>0$ in the above equation. A jump occurs at $t\in[0,L]$ provided
that $\{x(t-),x(t+)\}\times T(t)$ contains a link. We have
$T'(t+)=-T'(t-)$ in case that link is a double bar and $T'(t+)=T'(t-)$
in case it is a cross. We identify loops with identical support
and 
we occasionally abuse notation and identify a loop with the 
set of links it traverses.
The number of loops in a configuration $\omega$ is denoted 
$\mathcal L(\omega)$. The number of links in a configuration $\omega$ is denoted by $\#\om$. Similarly the number of double bars is denoted by $\#\om_{\dbar}$ and the number of crosses is denoted by $\#\om_{\cross}$.


For $u\in\mathbb{R}$, we define the following 
\emph{signed measure} on the
set $\Omega_{\ell,\beta}$
 of link configurations $\om$:
\be\label{eq:unnormalisedPPP}
\dd\bar\rho_{u}(\omega) = u^{\#\omega_{\cross}}
\dd^{\otimes \#\omega}x,
\ee
where $\dd x$ is the
Lebesgue measure on $\overline E_{\ell,\beta}$.
We also introduce the following normalized measure
$\rho_{u}$, satisfying 
$\rho_{u}(\Omega_{\ell,\beta})=1$:
\be\label{eq:normalisedPPP}
\dd\rho_{u}(\om)=\e{-(1+u) 2\beta |E_{\ell}|}
\dd\bar\rho_{u}(\om)
\ee
If $u$ is positive, 
the measure $\rho_{u}$ is a positive measure and hence a probability
measure;  in fact, under this measure
 $\om$ has the distribution of a Poisson point process 
with intensity $u$ for crosses $\cross$ and
intensity $1$ for double-bars $\dbar$.  
But we also allow small, negative $u$.  
Let 
\be \label{eq:Z-def}
Z_{\ell,\beta,n,u}\deq
\int_{\Omega_{\ell,\beta}} \dd\rho_{u}(\omega) \, n^{\mathcal L(\omega)-\#\omega_{\dbar}}.  
\ee 

This loop model is equivalent to the quantum spin system, and the next
result is an instance of this equivalence.  The equivalence goes back
to 
T\'oth \cite{Toth} 
and Aizenman--Nachtergaele \cite{AN} for special choices of the
parameters; 
the general  case of the interaction \eqref{hab}
is due to 
 \cite{Uel}.  Note that it holds for arbitrary finite graphs, not only for chains. 

We write $x \not\leftrightarrow y$
to characterize the set of
configurations $\omega$ where $(x,0)$ and $(y,0)$ belong to distinct
loops; 
$x\overset{+}{\longleftrightarrow} y$ where the top of $(x,0)$ is connected to the
bottom of $(y,0)$; and $x \overset{-}{\longleftrightarrow} y$  where the top of
$(x,0)$ is connected to the top of $(y,0)$ (see \cite[Fig.\ 2]{Uel}
for an illustration).

\begin{theorem}
\label{thm spin loop}
For the Hamiltonian \eqref{ham} with $h_{x,x+1}=-u T_{x,x+1}-Q_{x,x+1}$,
we have that
\begin{itemize}
\item[(a)] $\displaystyle \Tr \e{-2\beta H_\ell} = \e{2\beta (1+u)
    |E_\ell|} Z_{\ell,\beta,n,u}$.
\item[(b)] For all $1 \leq \alpha < \alpha' \leq n$, we have
\[
\Tr L_x^{\alpha,\alpha'} L_y^{\alpha,\alpha'} \e{-2\beta H_\ell} = \tfrac2n \e{2\beta (1+u) |E_\ell|} \int_{\Omega_{\ell,\beta}} \dd\rho_{u}(\omega) \, n^{\mathcal L(\omega)-\#\omega_{\dbar}} \bigl( \bbone[x \overset{-}{\longleftrightarrow} y] - \bbone[x\overset{+}{\longleftrightarrow} y] \bigr).
\]
\end{itemize}
\end{theorem}

The sign of the parameter $u$ in the definition of the interaction has indeed changed; but the theorem holds for arbitrary real (or even complex) parameters. Theorem \ref{thm spin loop} can also be formulated for the interaction $h_{x,x+1}=-u T_{x,x+1}-vQ_{x,x+1}$, by inserting the factor $v^{\#\omega_{\dbar}}$ inside the integrals.

\begin{proof}
The proof of (a) can be found in \cite[Theorem 3.2]{Uel} and (b) is similar, so we only sketch it here. Let $\Sigma(\omega)$ be the set of ``space-time spin configurations" that are constant along the loops (so that $|\Sigma(\omega)| = n^{\caL(\omega)}$). By a standard Feynman-Kac expansion, we get
\be
\Tr \e{-2\beta H_\ell} = \e{2\beta (1+u) |E_\ell|} \int_{\Omega_{\ell,\beta}} \dd\rho_{u}(\omega) \, n^{-\#\tilde\omega_{\dbar}} \sum_{\sigma \in \Sigma(\omega)} 1.
\ee
We recognize the partition function in \eqref{eq:Z-def}, so we get (a).

For (b) we need a modified set of space-time spin configurations where the spin value must jump from $\alpha$ to $\alpha'$, or from $\alpha'$ to $\alpha$, at the points $(x,0)$ and $(y,0)$. Let $\Sigma_{x,y}^{\alpha,\alpha'}$ be this set. We then have
\bm
\Tr _x^{\alpha,\alpha'} L_y^{\alpha,\alpha'} \e{-2\beta H_\ell} = \e{2\beta (1+u) |E_\ell|} \int_{\Omega_{\ell,\beta}} \dd\rho_{u}(\omega) \, n^{-\#\tilde\omega_{\dbar}} \\
\sum_{\sigma \in \Sigma_{x,y}^{\alpha,\alpha'}(\omega)} \langle \sigma_{x,0+} | L_x^{\alpha,\alpha'} | \sigma_{x,0-} \rangle \, \langle \sigma_{y,0+} | L_y^{\alpha,\alpha'} | \sigma_{y,0-} \rangle.
\end{multline}
It is necessary that $(x,0)$ and $(y,0)$ belong to the same loop in order to get a nonzero contribution. Further, we have
\be
 \langle \sigma_{x,0+} | L_x^{\alpha,\alpha'} | \sigma_{x,0-} \rangle \, \langle \sigma_{y,0+} | L_y^{\alpha,\alpha'} | \sigma_{y,0-} \rangle = \begin{cases} -1 & \text{if } x\overset{+}{\longleftrightarrow} y, \\ +1 & \text{if } x\overset{-}{\longleftrightarrow} y. \end{cases}
 \ee
Since $|\Sigma_{x,y}^{\alpha,\alpha'}(\omega)| = \frac2n n^{\caL(\omega)}$, we get (b).
\end{proof}

From now on and to the end of this article we work with the loop model.

\begin{rem}[Intuition] 
It is helpful to think of $\rho_{u}$   as
an a-priori measure on a gas of loops, and rewrite
the integrand 
$n^{\mathcal L(\omega)-\#\omega_{\dbar}}$
as  $\e{-(\log n) H(\omega)}$, with `Hamiltonian'
\be\label{eq:Hamiltonian}
-H(\omega)\deq \mathcal L(\omega)-\#\omega_{\dbar}, 
\ee 
and inverse temperature $\log n$. 
Thinking of $n$ as large, the
Laplace principle tells us that `typical'
configurations should maximise 
$n^{\mathcal L(\omega)-\#\omega_{\dbar}}$.  Our  goal is
to write $Z_{\ell,\beta,n,u}$ as a dominant contribution from
such maximizers, and some excitations.
\end{rem}

We end this section with the following remark about working with a
signed measure.
Since the (possibly signed) measure $\rho_{u}$ is closely related to
the probability measure $\rho_{1}$, it is easy to see that any event
$A$ satisfying $\rho_{1}(A)=0$ also has zero measure under
$\rho_{u}$.  In fact, we have the following slightly stronger
property: 

\begin{lemma}
If $A$ is an event such that $\rho_{1}(A)=0$ and 
$f:\Omega_{\ell,\beta}\to\mathbb{R}$ is a $\rho_{1}$-integrable
function, then for any $u\in\mathbb{R}$ we have that
\be
\int_A \dd\rho_{u}(\om) f(\om)=0.
\ee
\end{lemma}

\begin{proof}
Using \eqref{eq:unnormalisedPPP} and \eqref{eq:normalisedPPP}
it is easy to see that
\be
\left|\int_A \dd\rho_{u}(\om) f(\om)\right|
\leq C\int_A \dd\rho_{1}(\om) |f(\om)|=0,
\ee
for some finite constant $C$ depending only on $u,\ell,\beta$.
\end{proof}

As a consequence, we may assume that crosses and double-bars occur at different times, also when $u<0$ and the measure $\rho_u$ carries signs. We implicitly used this property when defining loops.

\section{The contour model}
\label{sec contours}

\subsection{Contours}

We classify loops as follows, see Fig.\ \ref{fig:loop-sorts}.
A loop is \emph{contractible} if it can be continuously deformed to a
point and \emph{winding} otherwise. 
Not all loops are contractible  since our time interval $T_\beta$ is
periodic. 
A loop is \emph{long} if it visits three or more distinct
vertices or if it is winding; it is \emph{short} otherwise.

\begin{figure}[htb]\center
  \includegraphics[scale=.6]{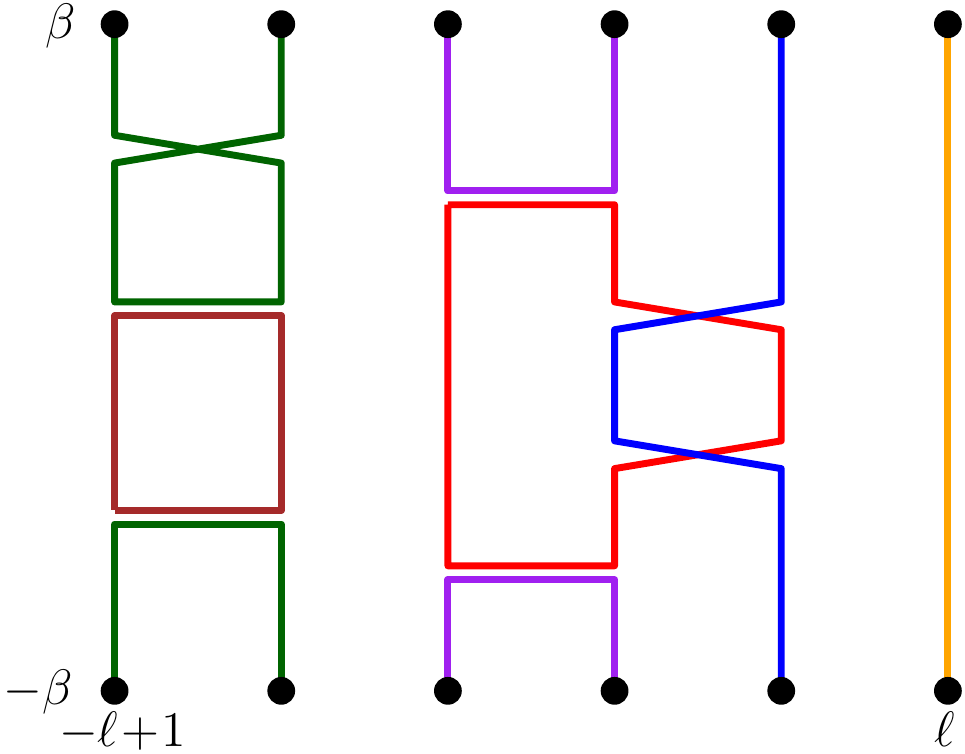}
  \caption{A configuration $\omega$ consisting of three short loops
    (green, brown, purple), and three long loops (red, blue, orange) two of which are winding loops
    (blue, orange). 
}\label{fig:loop-sorts}
\end{figure}

We define a \emph{canonical orientation} of the space-time vertices
$\overline V_{\ell,\beta}$, using the directions up ($\up$) and down ($\dn$),
by orienting the leftmost space-time
vertex  $\{-\ell+1\}\times T_\beta$ \emph{down} $\dn$ and
requiring that neighbouring space-time
vertices have opposite orientations;
see Fig.\ \ref{fig:orientations}.
We write $V^\up_\ell\deq \{x\in V_\ell: x+\ell \mbox{ is even}\}$
for the set of vertices with up-orientation,
and $V^\dn_\ell\deq \{x\in V_\ell: x+\ell \mbox{ is odd}\}$
for the set of vertices with down-orientation,
and introduce the following subsets of the edge-set $E_\ell$:
\be\begin{split}
E_\ell^+&\deq \big\{(x,x+1)\in E_\ell:
x\in V^\dn_\ell, x+1\in V^\up_\ell\big\},\\
E^-_\ell & \deq E_\ell\setminus E^+_\ell=
\big\{(x,x+1)\in E_\ell:
x\in V^\up_\ell, x+1\in V^\dn_\ell\big\}.
\end{split}
\ee 
We define  $\overline E_{\ell,\beta}^+$ and $\overline E_{\ell,\beta}^-$,  as well as
 $\overline V^\up_{\ell,\beta}$ and  $\overline V^\dn_{\ell,\beta}$,
analogously.

\begin{figure}[htb]\center
  \includegraphics[scale=.65]{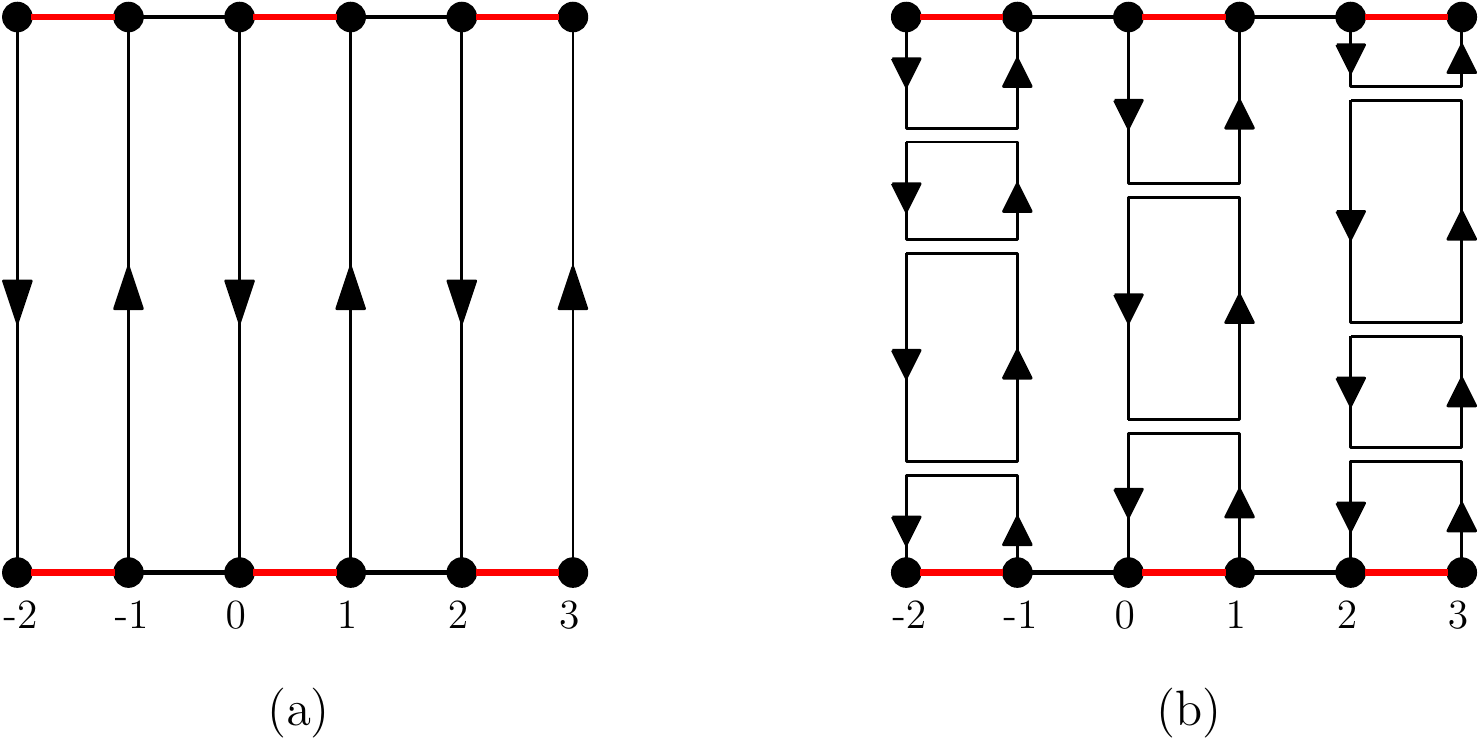}
  \caption{(a) The canonical orientation of $\overline V_{\ell,\beta}$
with the set $E^+_\ell$ highlighted red.  (b)
a configuration $\omega$ with many short loops;
these are positively oriented under the canonical orientation.
}\label{fig:orientations}
\end{figure}

These definitions are motivated as follows.  We expect that `typical'
configurations $\omega$  contain many short loops.  
To maximize the number of short 
loops one places only double-bars in $\overline E^+_\ell$, as in 
Fig.\ \ref{fig:orientations} (b).
The canonical orientation is chosen so that all the short loops in such a
configuration are positively oriented (i.e.\ 
counter-clockwise).  The canonical orientation will
be useful in classifying the excitations away from such `typical' 
$\omega$.
Also note that if the origin $0$ belongs to a short, positively
oriented loop, then we have $0\leftrightarrow 1$ for $\ell$ odd and 
$0\leftrightarrow -1$ for $\ell$ even.  To prove our main result
Theorem \ref{thm main}
we will essentially argue that  the origin is likely to belong to a
short, positively oriented loop.

Given a loop $l$ in a configuration $\omega$,
we define a  \emph{segment} of $l$ as a trajectory of
$l$ between two times $0\leq s_1<s_2\leq L(l)$ when $l$ 
 passes through height $\beta$.   That is to say,
$l(s_1)=(v_1,\beta), l(s_2)=(v_2,\beta)$ for some
$v_i\in V_\ell$, while $l$ does not pass 
through height $\beta$ in times $t\in(s_1,s_2)$. 
We say that a segment is \emph{spanning} if for every $t\in T_\beta$
there exists a $v=v(t)\in V_\ell$ such that the segment traverses
$(v,t)$. 
Note that a spanning segment is not necessarily part of a winding
loop. See Fig.\ \ref{fig:winding}. 

\begin{figure}[hbt]\center
  \includegraphics{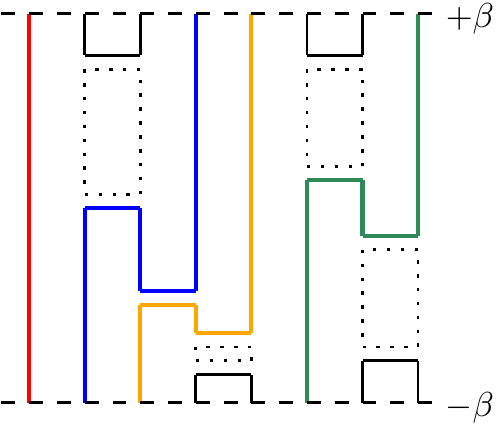}
  \caption{The leftmost and rightmost loops are winding loops with one spanning segment each. The loop in the middle is contractible. There are four spanning segments in total.}\label{fig:winding}
\end{figure}

\begin{defi}[Contours]
We say that two loops are \emph{connected} if they
share a link or both are winding.  
A \emph{contour} is then a maximally connected set of long loops.
\end{defi}


\begin{rem}\label{rk:cross}
For later reference, we note here that any cross $\cross$ which is
traversed by some loop in a contour is necessarily traversed both
ways by the contour;  see Fig.\ \ref{fig:cross-in-contour}.
\end{rem}
\begin{figure}[hbt]\center
  \includegraphics[scale=.6]{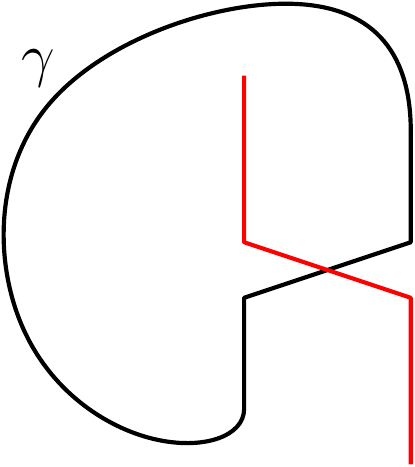}
  \caption{
A cross is traversed by a contour $\g$.  If the red loop visits a
third vertex, it is a long loop;  otherwise it must be a winding
loop.  In both cases, it is actually part of $\g$.
} \label{fig:cross-in-contour}
\end{figure}

A contour which contains at least one winding 
loop will be called a \emph{winding contour}. See Fig.\ \ref{fig:contours}.

\begin{figure}[htb]\center
  \includegraphics{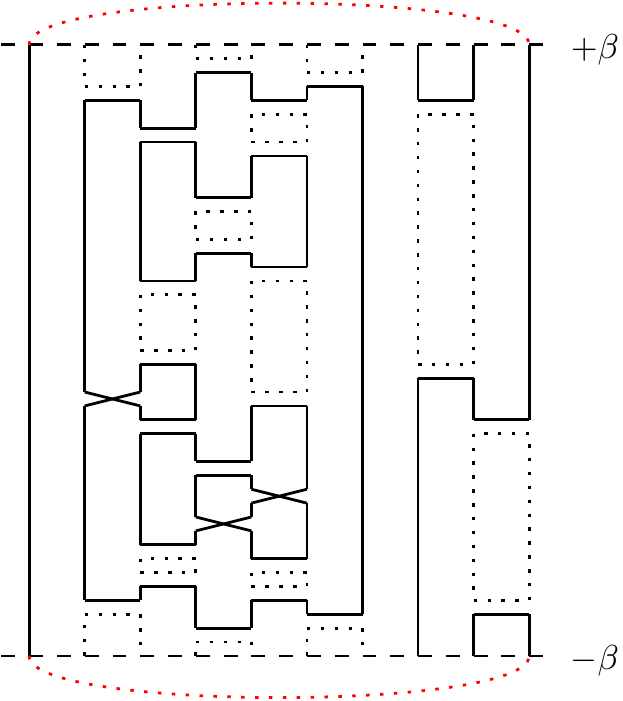}
  \caption{
Two contours:  One winding contour, consisting of two winding loops, and one
consisting of four long, but contractible loops.
} \label{fig:contours}
\end{figure}

We need a notion of  \emph{interior} of a contour, and for
this it is useful to regard our configuration $\omega$ as
living in the bi-infinite cylinder $C_\beta = \bbR \times T_\beta$.
More precisely, given $\omega$ we consider the subset 
$\overline{\omega}$
of $C_\beta$ obtained as the union of (i)  $\overline V_\beta$ embedded in 
$C_\beta$ in the natural way, and  (ii) the links of $\omega$
embedded as straight line segments connecting adjacent points of 
$\overline V_\beta$.  Note that, in the embedding $\overline\omega$,
crosses and double-bars are embedded in the same way.  For a loop $l$
of $\omega$,
define its \emph{support} $S(l)$ as the subset of $\overline\omega$
traced out by $l$, meaning the union of the vertical and horizontal
line segments of $\overline\omega$ corresponding to the intervals of 
$\overline V_\beta$ and the links of $\omega$ traversed by $l$.
For a contour $\gamma$ of $\omega$ we then make the following
definitions.
\begin{itemize}
\item The \emph{support} $S(\gamma)$ is 
 the union of the supports  $S(l)$ of the loops $l$ belonging to
 $\gamma$.   Note that $S(\gamma)$ is a closed subset of $C_\beta$.
\item The \emph{exterior} $E(\gamma)$ is the union of the unbounded
  connected components of
  $C_\beta\setminus S(\gamma)$.  Note that   
$E(\gamma)$ is open.
\item The \emph{interior} 
$I(\gamma)\deq {C_\beta\setminus \ol{E(\gamma)}}$.
Note that $I(\gamma)$ is an open set.
\item The \emph{boundary} 
$B(\gamma)\deq\overline{E(\gamma)}\setminus E(\gamma)$
which is a closed set.
\item The (vertical) {\it length} $|\gamma|$ of a contour as the sum of the (vertical)
lengths of its loops, $|\gamma|\deq\sum_{l\in \gamma}|l|$.

\end{itemize}
These notions are illustrated in Figs \ref{fig:interior1}--\ref{fig:interior3}.

Having defined $I(\gamma)$ as a subset of the cylinder $C_\beta$, 
we may also regard $I(\gamma)$ (or more precisely, 
its closure $\ol{I(\gamma)}$) as a subset of 
$\overline E_{\ell,\beta}$ by identifying a point 
$(x,x+1)\times\{t\}\in\overline E_{\ell,\beta}$ with the
closed line-segment from $(x,t)$ to $(x+1,t)$ in $C_\beta$.
Similarly, $S(\gamma)$ and $B(\gamma)$ may be
regarded as subsets of $\overline V_{\ell,\beta}\cup\omega$.
 We freely switch between these 
points of view. 

\begin{figure}[htb]\center
  \includegraphics[scale=.6]{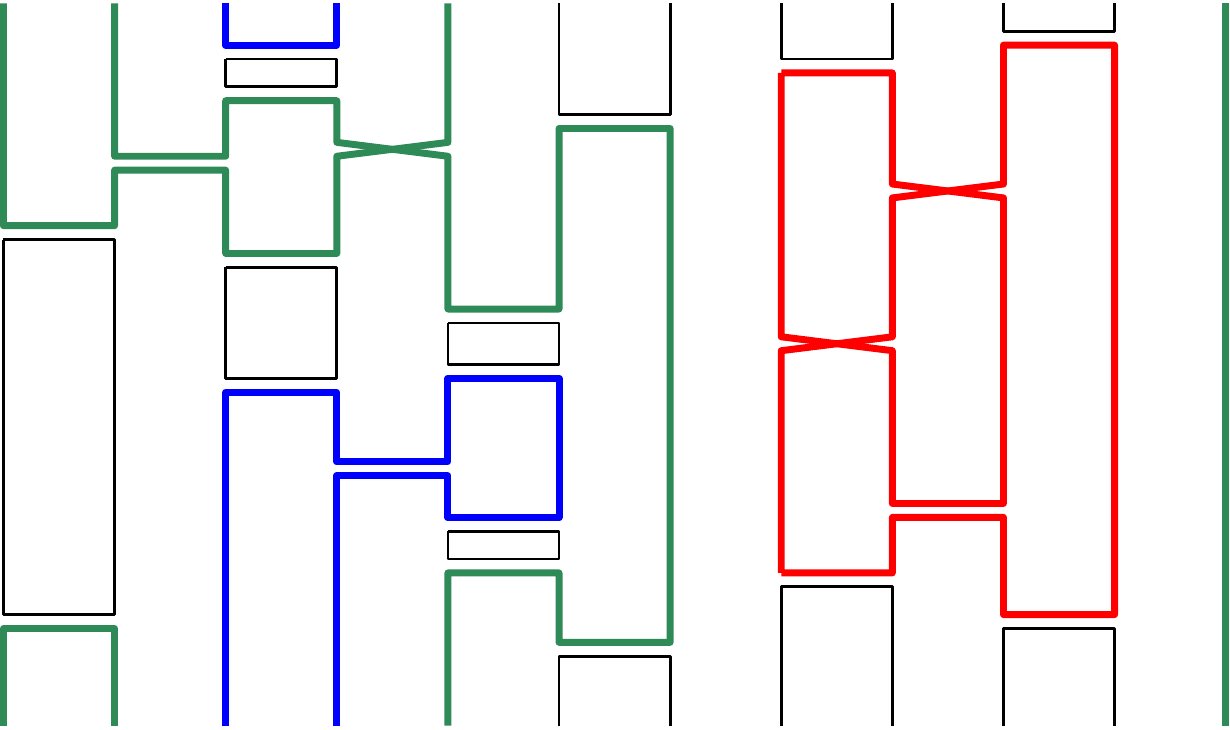}
  \caption{A configuration $\omega$ with three contours highlighted
    green, blue and red.  The green contour consists of two winding loops.}\label{fig:interior1} 
\end{figure}

\begin{figure}[htb]\center
  \includegraphics[scale=.6]{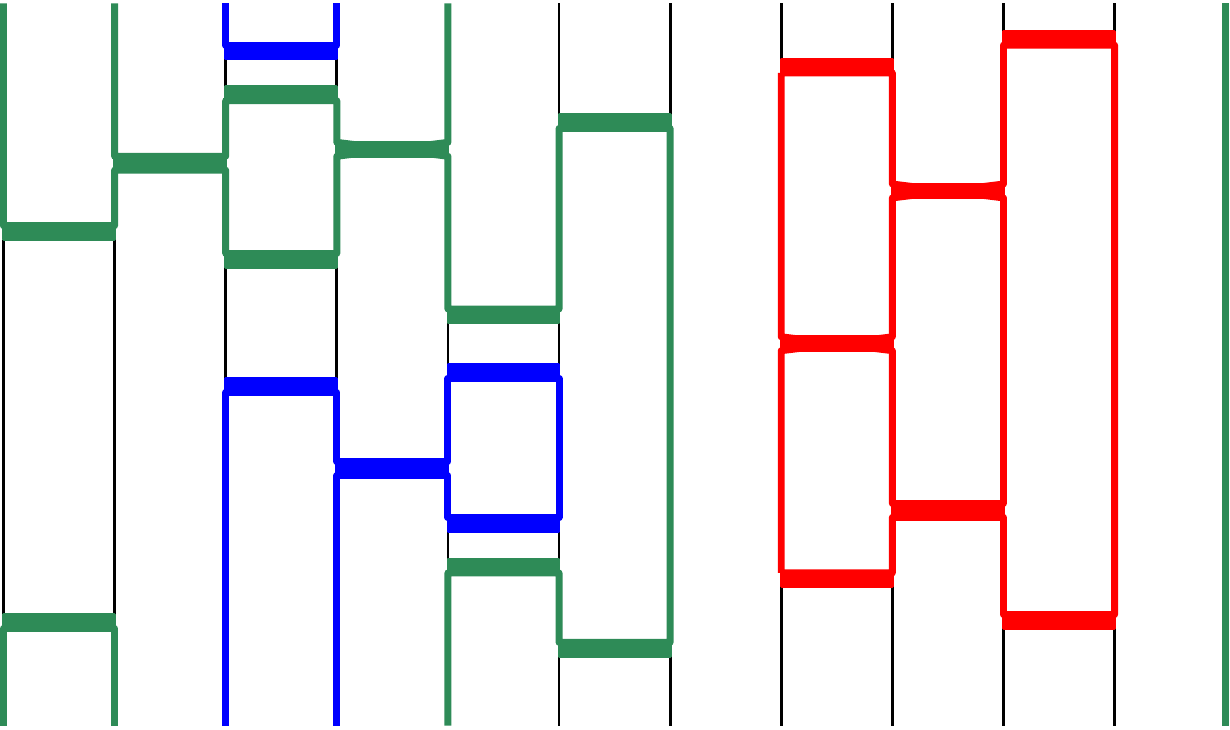}
  \caption{The corresponding embedding 
$\overline\omega\subseteq C_\beta= T_\beta\times \mathbb{R}$,
with the supports $S(\gamma)$ of the contours highlighted with the
corresponding colors.}\label{fig:interior2}
\end{figure}

\begin{figure}[htb]\center
  \includegraphics[scale=.6]{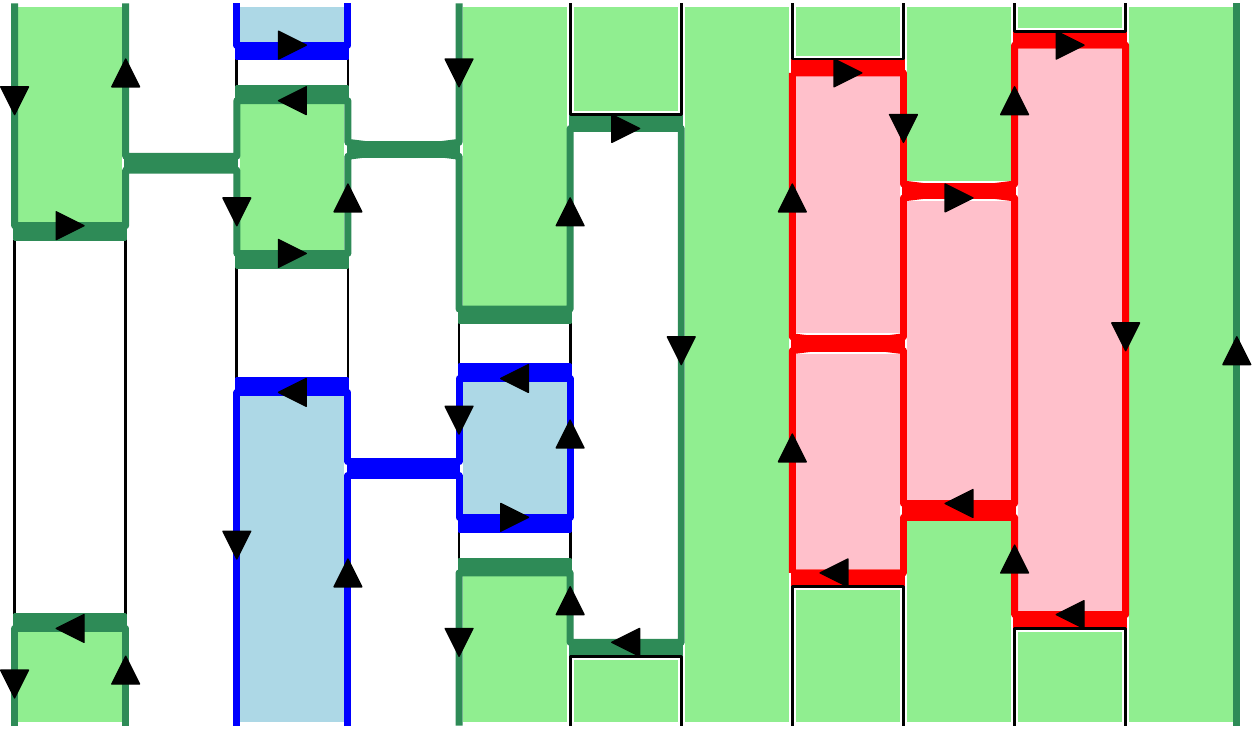}
  \caption{The interiors of the corresponding contours with the
    boundaries $B(\gamma)$ receiving the canonical orientation.
The green and blue contours are of positive type
 (interiors $I(\gamma)$  on the left)
 while the
 red contour is of negative type
 (interior  on the right).  
}\label{fig:interior3}
\end{figure}

Fixing a contour $\gamma$, note that
the boundary $B(\gamma)$ consists of a collection of closed
curves and horizontal line segments (of length
1).  We use the canonical orientation of $\overline V_{\ell,\beta}$ 
to orient each vertical segment of  $B(\gamma)$.
It is not hard to see that this gives a consistent orientation of all
the closed curves constituting $B(\g)$.
(This follows from Remark \ref{rk:cross}.)
Recall the standard notion of a 
\emph{positively oriented curve} as one whose
interior is always on the left.

\begin{defi}[Type of a contour]
 We say that the contour $\gamma$ is of \emph{positive type}  if 
the canonical orientation of $B(\gamma)$ is \emph{positive} in the
sense that $I(\g)$ is on the left of each closed curve of $B(\g)$.
Otherwise we say that $\gamma$ is of \emph{negative type}  
(being of negative type  is equivalent to the interior being on the right).
\end{defi}

\begin{rem}\label{rk:pos}
Suppose that $\om\in\Omega_{\ell,\beta}$ is such that a given point 
$\bar v\in\overline V_{\ell,\beta}$ is not on or inside any contour, that
is to say
\be
\bar v\in\bigcap _{\g\in\G(\om)} E(\g),
\ee
where $E(\g)$ is the exterior of $\g$
defined above.  Then we have that $\bar v$ is
on a \emph{positively oriented} short loop.  Indeed, this is related
to the fact that all external contours are of positive type, see Lemma
\ref{lem:admissibility}.
\end{rem}

\subsection{Domains and admissibility of contours}

We now introduce several notations and definitions pertaining to
contours and how they relate to each other.  First, given
$\om\in\Omega_{\ell,\beta}$ we define 
$\Gamma(\om)=\{\g_1,\dotsc,\g_k\}$ as the set of
contours in the configuration $\om$.  
Here, and in what follows, a contour may be identified with the set of
links it traverses.
The collection of all possible
contours will be denoted 
$X_{\ell,\beta}=\bigcup_{\om\in\Omega_{\ell,\beta}} \Gamma(\om)$,
and we write $X^+_{\ell,\beta}\se X_{\ell,\beta}$ for the collection
of \emph{positive-type} contours.  We write 
\be
\fX_{\ell,\beta}=\bigcup_{k\geq 0} \binom{X_{\ell,\beta}}{k}
\quad\mbox{and}\quad
\fX^+_{\ell,\beta}=\bigcup_{k\geq 0} \binom{X^+_{\ell,\beta}}{k}
\ee
for the set of finite collections of contours, respectively
positive-type contours.  Elements of $\fX_{\ell,\beta}$ and of
$\fX^+_{\ell,\beta}$ will usually be denoted by $\Gamma$.   It is
important to note that far from every such set $\Gamma$ of 
contours can be obtained as
$\Gamma(\om)$ for some $\om\in\Omega_{\ell,\beta}$;  in fact, we will
devote some effort to identifying criteria under which such an
$\omega$ does indeed exist.  We say that
$\Gamma\in\fX_{\ell,\beta}$ is \emph{admissible} if
$\Gamma=\Gamma(\om)$ for some $\om\in\Omega_{\ell,\beta}$,
and write $\fA_{\ell,\beta}=\Gamma(\Omega_{\ell,\beta})$ for the
collection of admissible sets of contours.

Recall that the interior $I(\gamma)$ of a contour $\gamma$ is by
definition an open subset of the cylinder $C_\beta$.
Also recall that we regard $\overline E_{\ell,\beta}$ as a closed subset of
$C_\beta$ by identifying a point
$(x,x+1)\times\{t\}\in\overline E_{\ell,\beta}$ with the
closed line-segment from $(x,t)$ to $(x+1,t)$.
We now define the \emph{(interior) domains} of $\gamma$ as follows.
\begin{defi}
A domain $D$ of $\g$ is a subset of 
$\overline E_{\ell,\beta}\cap {I(\gamma)}$ which, when regarded as a subset
of $C_\beta$ as above, is connected, satisfies 
$D\cap S(\gamma)=\varnothing$,
and is maximal with these properties.
\end{defi}
We define the \emph{type} of a domain in a similar way to the
type of a contour.  Namely, 
we orient the (topological) boundary of $D$ consistenly with the
canonical orientation of $\overline V_{\ell,\beta}$
and say that $D$ is of \emph{positive type} if this
is a positive orientation (interior on the left), 
and of \emph{negative type} otherwise.
See Fig.\ \ref{fig:domains}.

\begin{figure}[hbt]\center
  \includegraphics{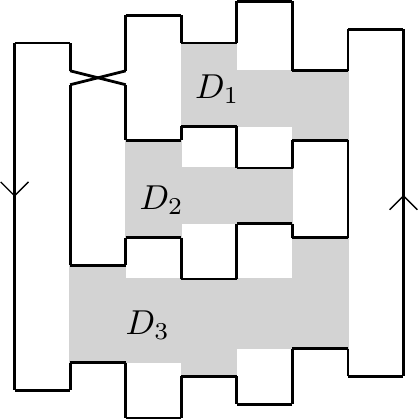}
  \caption{A contour $\g$ of positive type, 
containing three domains $D_1,D_2,D_3$.
Domains $D_1$ and $D_3$ are of negative 
type, 
while $D_2$ is of positive type. 
}\label{fig:domains}
\end{figure}

Given two contours $\g$ and $\g'$, we say that
$\g$ is a  \emph{descendant} of  $\g'$,
writing $\g\prec\g'$,
 if $S(\g)\se D$ for some domain $D$ of $\g'$.
Given $\Gamma\in\fX_{\ell,\beta}$ and 
$\g,\g'\in \Gamma$, 
 we say that $\g$ is an \emph{immediate descendant 
of $\g'$ in $\Gamma$} if $\g\prec\g'$
and there is no $\overline\g\in\Gamma$ satisfying 
both $\g\prec\overline\g$ and $\overline\g\prec\g'$.
It is important to note that the notion of being an 
\emph{immediate} descendant depends not only on the two contours
$\g$ and $\g'$ but on the  set $\Gamma$; in other
words, immediate descendancy cannot be checked in a 
pairwise manner.
If $\g\in\Gamma$ is not the descendant of any other contour
$\g'\in\Gamma$ then we say that $\g$ is an \emph{external} contour;
this notion is also dependent on the set $\Gamma$.

Note that the unique (if it exists) winding contour is always external since a winding loop cannot be in the interior of any contractible loop. 

\begin{lem}\label{lem:admissibility}
Fix $\Gamma\in\fX_{\ell,\beta}$.  Then $\Gamma$ is admissible, i.e.\
$\Gamma\in\fA_{\ell,\beta}$,  
if and only if the following hold:
\begin{enumerate}
\item all external contours in $\Gamma$ are of positive type;
\item  for any pair of distinct contours 
$\g,\g'\in\Gamma$ we have that
 either $\ol{I(\g)}\cap\ol{I(\g')}=\varnothing$ or $\g\prec\g'$ or
  $\g'\prec\g$; 
\item if $\g$ is an immediate descendant of $\g'$, in a domain $D$ of
  $\g'$, then the types of $\g$ and of $D$ coincide;
\item there exists at most one winding contour $\g\in\G$.
\end{enumerate} 
\end{lem}
\begin{proof}
It is easy to see that the four conditions above hold for any
admissible $\G=\G(\om)$.
To show the converse, we  construct an explicit 
$\omega\in\Omega_{\ell,\beta}$ with $\Gamma(\omega)=\Gamma$. 
Starting from the empty configuration 
$\omega_0=\varnothing\in\Omega_{\ell,\beta}$, add all links of all
external contours and then place a double bar at height 0, say, on each
$e\in E_\ell^+$ that does not have any link on it.  This defines
$\om_1$ such that $\G(\om_1)$ is precisely the set of external
contours of $\G$.
Next, add the links of all contours which are immediate descendants of
external contours.   This does not
create any new long loops apart from those in these contours because
their types coincide with those of the domains they are in.
 Iterate this procedure until there are no more contours left to add.
\end{proof}

An important prerequisite for applying a cluster expansion is to
be able to verify the admissibility of a set of contours in a pairwise
manner.  As indicated above, and in the light of Lemma
\ref{lem:admissibility}, this is not directly possible since the
notion of being an immediate descendant depends on the whole set
$\G$.  We get around this issue by introducing a notion of
\emph{compatibility} which applies to sets of positive-type contours
$\G\in\fX^+_{\ell,\beta}$, and which can be checked in a pairwise
manner.  We then show that there is a bijective correspondence between
admissible and compatible sets of contours.

The bijective correspondence referred to above involves
\emph{shifting} contours and rests on the simple observation that if
$\g$ is a negative-type contour, then $\g'=\g+(1,0)$ (i.e.\ $\g$
translated to the right one unit) is a positive-type contour.

\begin{figure}[hbt]\center
  \includegraphics[width=5cm]{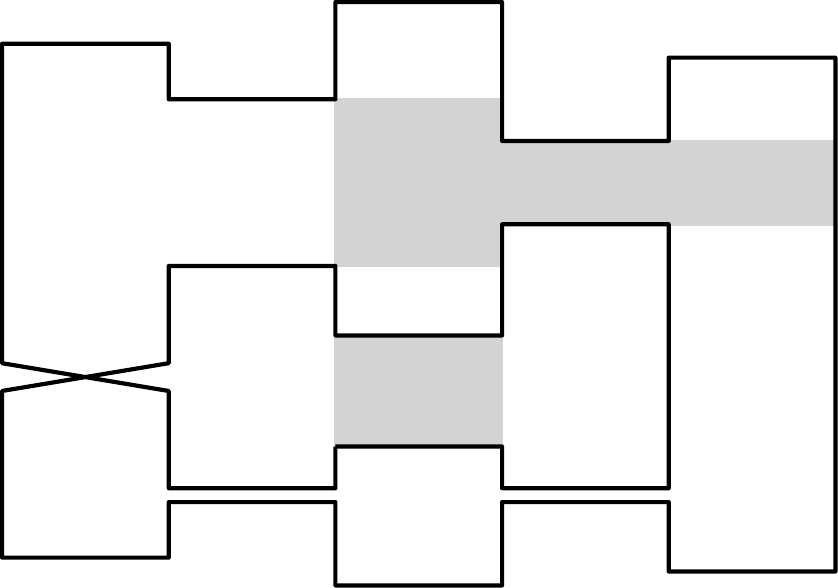}
  \caption{A contour $\gamma$ and its two appropriately shifted
    domains (shaded areas). The lower one was not moved since it
    already was positive type. The upper one was shifted one column to the
    right. If a $\gamma'\in X_{\ell,\beta}$ gets placed inside it,
    $S^{-1}(\Gamma=(\gamma,\gamma'))$ will return an admissible
    collection of contours.}
\label{fig:shifted-domains}
\end{figure}

Given a positive-type contour $\g\in X^+_{\ell,\beta}$ with domains 
$D_1(\g),\dotsc,D_k(\g)\se I(\g)$, we define the 
\emph{appropriately shifted domains} $D_i^+(\g)$ of $\g$ by 
\be
D_i^+(\g)=\left\{
\begin{array}{ll}
D_i(\g), & \mbox{if } D_i(\g) \mbox{ is of positive type}, \\
D_i(\g)+(1,0), & \mbox{otherwise}.
\end{array}
\right.
\ee
Note that while $D_i^+(\g)\se\ol{I(\g)}$, a shifted domain may
intersect the boundary $B(\g)$.
See Fig.\ \ref{fig:shifted-domains}.

\begin{defi}\label{def:compatible}
Given two positive-type contours $\g,\g'\in X^+_{\ell,\beta}$,
we say that $\g$ and $\g'$ are \emph{compatible} if one of the
following hold:
\begin{enumerate}
\item $\ol{I(\g)}\cap \ol{I(\g')}=\varnothing$, or
\item $S(\g)\se D_i^+(\g')$ for some $i$, or
\item $S(\g')\se D_i^+(\g)$ for some $i$, or
\item at least one of $\g$ or $\g'$ is not a winding contour.
\end{enumerate}
We define
\be
\delta(\g,\g') = \begin{cases} 1 & \text{if $\g,\g'$ are compatible} \\ 0 & \text{otherwise.} \end{cases}
\ee
Finally, we let $\fC^+_{\ell,\beta}\se \fX^+_{\ell,\beta}$ to be the collection of all
pairwise compatible sets of positive-type contours; that is,
$\G=\{\g_1,\dotsc,\g_k\}\in\fX^+_{\ell,\beta}$ belongs to 
$\fC^+_{\ell,\beta}$ if $\prod_{1\leq i<j\leq k} \delta(\g_i,\g_j)=1$.
\end{defi}

A compatible set $\G$ is generally itself not admissible, since compatible contours may overlap but admissible contours may not.
Intuitively, one obtains an admissible set  of contours from
a compatible set by `shifting back' the appropriately shifted domains
$D_i^+(\g)$ and the contours they contain.  For nested contours, the
shift is performed iteratively; see Fig.\ \ref{fig:shifting-back}.

\begin{figure}[htb]\center
\includegraphics[scale=.65]{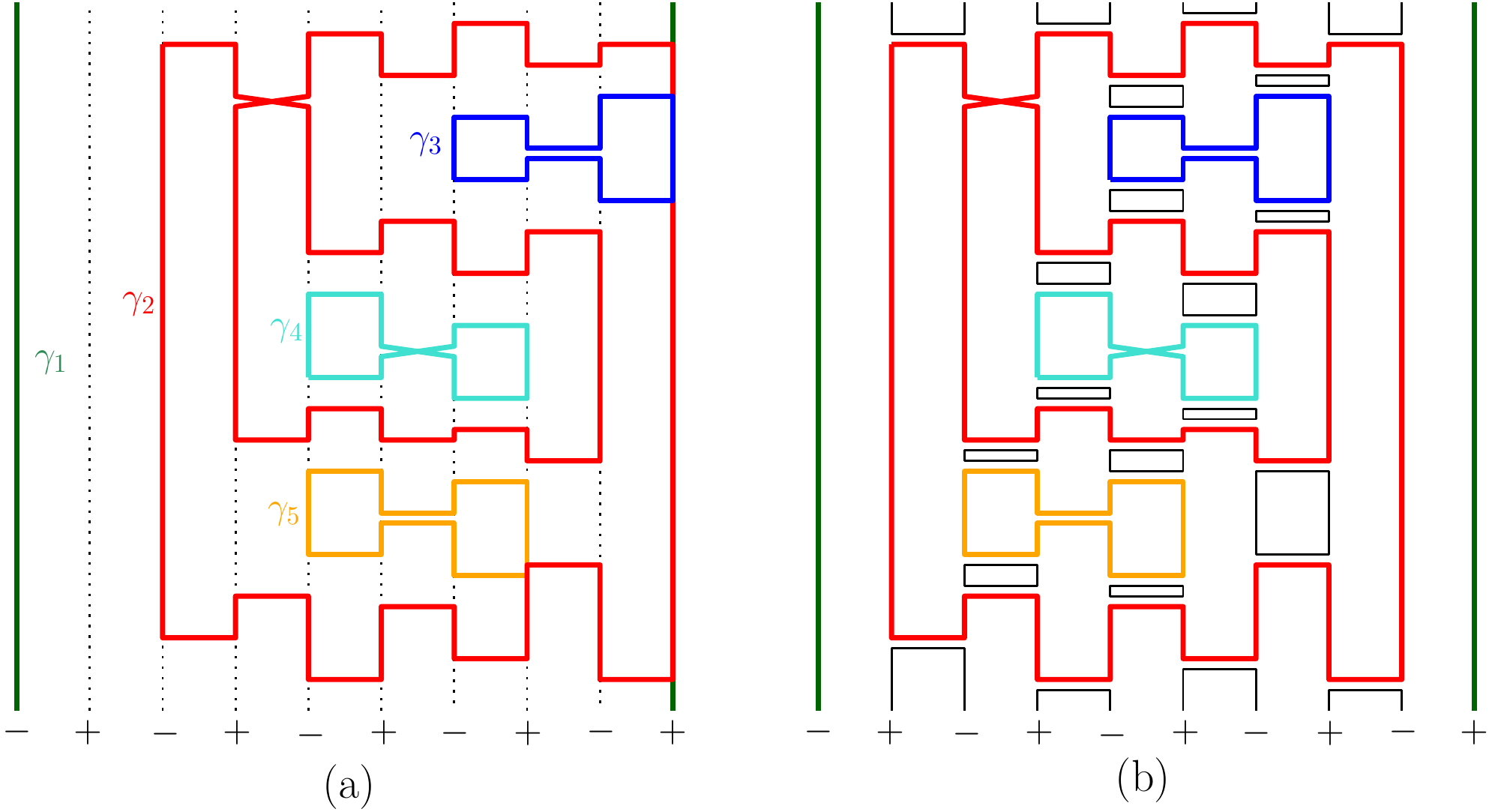}
\caption{(a) A compatible set of contours 
$\G=\{\g_1,\g_2,\g_3,\g_4,\g_5\}$.  We have, for example,
$\s_\G(\g_3)=2$ since $\g_3$ lies in a shifted domain of both $\g_1$
and $\g_2$, while $\s_\G(\g_4)=1$ since it lies in a shifted domain of
$\g_1$ only.  (b) The admissible set $\S(\G)$.
} 
\label{fig:shifting-back}
\end{figure}

 More formally,
define the \emph{shift}
$\S:\fC^+_{\ell,\beta}\to \fX_{\ell,\beta}$ as follows.  First, given 
$\G\in \fC^+_{\ell,\beta}$ and $\g\in\G$, write 
$\s_\G(\g)$ for the number of contours $\g'\in\G\setminus\{\g\}$
such that $\g\se D^+_i(\g')\neq D_i(\g')$.  
This represents the number of times $\g$ is shifted to the right in order
to obtain the compatible set $\G$ from an admissible set of contours.  
We define
\be
\S(\G)=\{\g-(\s_\G(\g),0):\g\in\G\}.
\ee

\begin{lem}\label{lem:pairwiseComp}
The shift $\S$ is a bijection from $\fC^+_{\ell,\beta}$, the
collection of compatible sets of contours, to $\fA_{\ell,\beta}$,
the collection of admissible sets of contours. 
\end{lem}
	
\begin{proof} 
It is easy to construct an inverse $\S^{-1}$ of $\S$ on
$\fA_{\ell,\beta}$, as follows.  Given $\G\in\fA_{\ell,\beta}$,
start with an external contour $\g$ (which is of positive type by
Lemma \ref{lem:admissibility}) and form its appropriately shifted
domains $D_i^+(\g)$.  In doing so, shift also the descendants of $\g$
along with their domains.  Note that all the immediate descendants of
$\g$ are then mapped to positive type contours.  Then iteratively
continue this procedure for the (shifted) immediate descendants of
$\g$.  The resulting set $\S^{-1}(\G)$ then satisfies 
Definition \ref{def:compatible}.  

It remains to show that $\S(\G)\in\fA_{\ell,\beta}$ for all
$\G\in\fC^+_{\ell,\beta}$, i.e.\ that $\S(\G)$ satisfies Lemma
\ref{lem:admissibility}. Compatibility ensures that there is at most one winding contour. It is clear that external contours are of
positive type since they are not shifted.  For $\g,\g'$ with disjoint
interiors, this property is preserved by $\S$;  if
$\g\se D_i^+(\g')$ then the shifting ensures that the images of
$\g,\g'$ under $\S$ satisfy $\g\prec\g'$, while the relative amounts
by which the contours are shifted ensures that the types of 
immediate descendants in $\S(\G)$ coincide with the types of the
relevant domains.
\end{proof}

We close this subsection with a simple lemma about 
counting the amount of
`available space' for  short loops in a configuration $\om$, in terms
of the lengths of the contours.  For $\G\in\fA_{\ell,\beta}$,
we define the \emph{free set} $F(\Gamma)\subseteq\overline E_{\ell,\beta}$ 
as the space-time edges 
where we can add links without modifying the contours in $\G$
or creating new ones.
\begin{lem}\label{lem:freeEdges} 
Let $\G\in\fA_{\ell,\beta}$ be an admissible set of contours. Then
\be\label{eq:FXeq} 
|F(\Gamma)|=|\overline E_{\ell,\beta}^+|
-\tfrac{1}{2}{\textstyle\sum_{\gamma\in\Gamma}|\gamma|.  }
\ee
\end{lem}
\begin{proof}
We need to show that
$2|\overline E_{\ell,\beta}^+|=2|F(\Gamma)|+
{\textstyle \sum_{\gamma\in\Gamma}|\gamma|.}$
Note that $2|\overline E_{\ell,\beta}^+|=|\overline V_{\ell,\beta}|$ and that 
$2|F(\Gamma)|$ equals the total length of all the short loops.
But any point in $\overline V_{\ell,\beta}$ lies either on a contour or on a short
loop, thus 
$|\overline V_{\ell,\beta}|=2|F(\Gamma)|+
{\textstyle \sum_{\gamma\in\Gamma}|\gamma|},$
as required.
\end{proof}

\subsection{Decomposition of $H(\omega)$}

Recall from \eqref{eq:Hamiltonian} the quantity
$-H(\omega)=\mathcal L(\omega)-|\omega_{\dbar}|$.
We now show that $H(\om)$
can be decomposed as a sum over
contours and we prove bounds on the summands.
To this end, for a loop $l$ let $\cT(l)$ denote the number of
\emph{turns} that $l$ makes;  symbolically
$\cT(l)=\#\dbartop+\#\dbarbot$.  For a contour $\g$, write $\cT(\g)$
for the total number of U-turns of all loops in $\g$.
Next define the function $h:X_{\ell,\beta}\to\mathbb{Z}$
by 
\be\label{eq:h-def}
h(\g)=\mathcal{L}(\g)-\tfrac12 \mathcal{T}(\g)
\ee
where $\mathcal{L}(\g)$ denotes the number of loops in the contour $\g$.

\begin{lem}\label{lem:Hh}
For $\omega\in\Omega_{\ell,\beta}$ with contours 
$\G=\Gamma(\om)$ we have
$-H(\omega)=\sum_{\gamma\in\Gamma}h(\gamma)$.
\end{lem}

\begin{proof}
Since every double-bar of $\om$ accounts for 
exactly two turns (of either one
or two loops), we have
\be
-H(\om)=\sum_l \big(1-\tfrac12\cT(l)\big)
\ee
where the sum is over all loops $l$ in the configuration $\om$.
The result now follows from the observation that short,
non-winding loops make exactly two turns.
\end{proof}

Write $\#\gamma_{\dbar}$ for the number of double-bars visited by $\g$
and $\#\gamma_{\cross}$ for the number of crosses.

\begin{lem}\label{lem:wBnds}
For contours $\g$ without crosses, the function
$h:X_{\ell,\beta}\to\mathbb{Z}$  satisfies 
\be\label{wBnds-rhs}
h(\gamma)\leq -\tfrac13\#\gamma+
2\ell \one\{\gamma\text{ has a spanning segment}\}.
\ee
\end{lem}
Note that the constant $-\tfrac13$ is tight for the smallest
non-winding contours with six double-bars and no crosses,
while for larger contours the constant may be taken closer to
$-\tfrac12$. 
As to the indicator function,
we will see that contours containing spanning
segments become very rare asymptotically.
\begin{proof}
Write  $\cW(l)$ for the number of winding segments in $l$
and $\cW(\g)=\sum_{l\in\g}\cW(l)$.   We claim that
it suffices to show that $h(\g)\leq r(\g)$ where 
\be
r(\g)=-\tfrac13\cT(\g)
+\cW(\g).
\ee
Indeed, $r(\g)$ is bounded above by the right-hand-side of 
\eqref{wBnds-rhs} for the following reasons:
\begin{itemize}
\item double-bars visited twice by $\g$ count twice in $\cT(\g)$ but
  only once in $\#\g=\#\g_{\dbar}$, while those visited once by $\g$ count
  once in both, meaning that $\cT(\g)\geq \#\g$;
\item $\cW(\g)\leq 2\ell \one\{\gamma\text{ has a spanning segment}\}$
since each point of the form $(x,0)\in\overline V_{\ell,\beta}$ is visited
by at most one winding segment.
\end{itemize}
Next, the claimed inequality $h(\g)\leq r(\g)$ is equivalent to:
\be\label{eq:ineq1}
\cT(\g)
+6\cW(\g)\geq 6\cL(\g).
\ee
To establish \eqref{eq:ineq1},
first note that both sides are additive over loops.  Thus it suffices
to show that any long or winding loop $l$ 
satisfies
\be\label{eq:ineq3}
\cT(l)
+6\cW(l)\geq 6.
\ee
If $\cW(l)\geq1$ this is clear, hence we may assume that the
loop is non-winding.  
A long, non-winding loop which traverses
only double-bars necessarily makes at least 6 turns, see Fig.\ \ref{fig:small-long-loop}.
This proves \eqref{eq:ineq3} and hence
the claim.
\end{proof}

\begin{lem}\label{lem:h leq 0}
	For all contours $\g\in X_{\ell,\beta}$ we have 
	\be
		h(\g)\leq 2\ell\one\{\gamma\text{ has a spanning segment}\}. 
	\ee
	In particular $h(\g)\leq 0$ for all non-winding contours.
\end{lem}
\begin{proof}
	Note that $h(\g)$ is additive over loops $l\in\g$. So since there can be at most $2\ell$ spanning segments, it suffices to show for every loop that $1-{1\over2}(\#\dbartop+\#\dbarbot)\leq \one\{l\text{ has a spanning segment}\}$. This is clearly true.
\end{proof}

\begin{lem}\label{lem:nice new lemma}
	For $\g\in X_{\ell,\beta}$, all $u$ with $|u|\leq 1$, and all $\kappa>0$, we have
\be
n^{h(\g)}|u|^{\#\g_{\cross}} \leq \min(n,|u|^{-\kappa/2})^{-(\frac13 - \kappa) \#\g} n^{2\ell{\SMALL\one}\{\gamma\text{ has a spanning segment}\}}.
\ee
\end{lem}

\begin{proof}
If $\g$ has no spanning segment and $\#\g_{\cross} \geq {\kappa\over2} (\frac13-\kappa) \#\g$, the claim follows from $h(\g) \leq 0$ (Lemma \ref{lem:h leq 0}). Now consider the case where $\g$ has no spanning segment and $\#\g_{\cross} \leq {\kappa\over2} (\frac13-\kappa) \#\g$. If $\#\g<({\kappa\over2}({1\over3}-\kappa))^{-1}$, then $\#\g_{\cross}=0$ and thus we may apply Lemma \ref{lem:wBnds} to get the desired bound in this case. So assume now that $1\leq{\kappa\over2}({1\over3}-\kappa)\#\g$ (and still that $\g$ has no spanning segment and $\#\g_{\cross} \leq {\kappa\over2} (\frac13-\kappa) \#\g$). 

Let  $\G$ denote the collection of
contours and small loops obtained by removing all crosses from
$\g$, and let $m$ denote the number of small loops in $\G$.
Since the removal of a cross can only create at most one more 
loop, we have that $m\leq\#\g_{\cross}+1\leq {\kappa\over2} (\frac13-\kappa) \#\g+1\leq \kappa (\frac13-\kappa) \#\g$,
and that
	\be\label{eq:h of gamma leq h gamma bar}
		h(\g)\leq h(\G)+\#\g_{\cross}\leq h(\G) + \frac{\kappa}{2} (\tfrac13-\kappa) \#\g.
	\ee
	Since every short loop uses at most two double bars, the
        number of double bars belonging to contours of $\G$ is
at least
\be
\#\g_{\dbar} - 2m \geq \bigl( 1 - {\kappa\over2} (\tfrac13-\kappa) \bigr) \#\g - 2 \kappa (\tfrac13-\kappa) \#\g = \bigl( 1-\tfrac52\kappa(\tfrac13-\kappa) \bigr) \#\g.
\ee
Applying Lemma \ref{lem:wBnds}, this allows us to conclude that 
	\be\label{eq:h of gamma bar}
		h(\G) \leq -\tfrac13 \big(1 - \tfrac52\kappa (\tfrac13-\kappa) \big) \#\g.
	\ee
	Combining \eqref{eq:h of gamma leq h gamma bar} 
and \eqref{eq:h of gamma bar} we conclude that 
\be
h(\g) \leq -\big[ \tfrac13 -3\kappa (\tfrac13-\kappa) \big] \#\g
\leq -(\tfrac13-\kappa) \#\g.
\ee
It remains to show the claim for $\g$ with a spanning segment. To this end, consider $\om$ obtained as follows: Denote by $\om_0$ a configuration of links such that its set of contours $\G(\om_0) = \{\g\}$. Now add $|E^+_\ell|=\ell$ double bars at the same height, exactly one per column in $E^+_\ell$. Denote this configuration by $\om$ and note that $\G(\om)$ does not contain any winding contours. Observe that the number of crosses is unchanged and we added $\ell$ links, hence changing the number of loops by at most $\ell$. Using these observations and Lemma \ref{lem:Hh} we thus get 
\be
h(\g)=-H(\om_0)=-H(\om)+\mathcal E = \sum_{\g'\in\G(\om)}h(\g') + \mathcal E,
\ee
where $\mathcal E$ is an error that is bounded by $|\mathcal E|\leq 2\ell$ and all $\g'\in\G(\om)$ are non-winding so that the previous bounds apply. This concludes the proof. 
\end{proof}

\begin{figure}[htb]\center
\includegraphics[scale=.5]{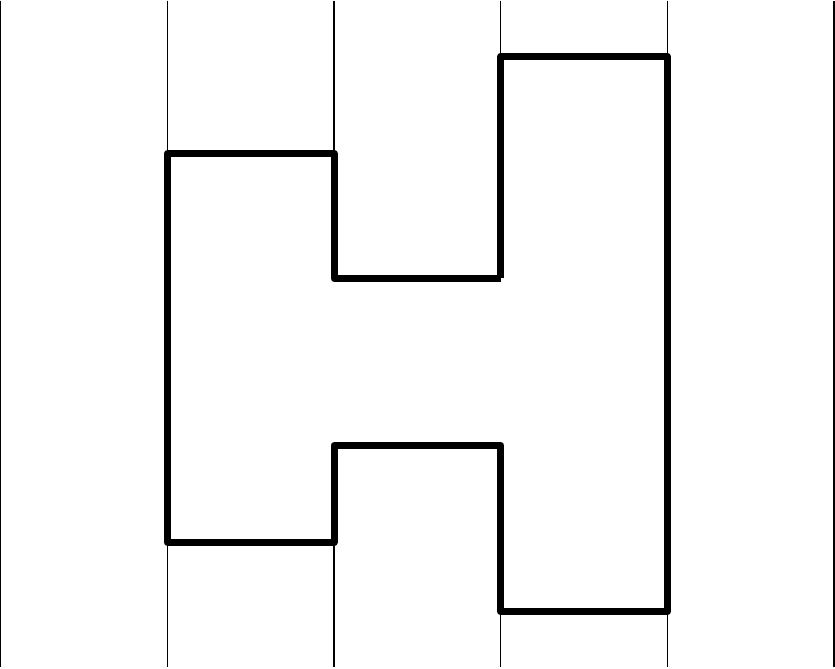}
\caption{
A long, non-winding loop makes at least 6 turns.
}\label{fig:small-long-loop}
\end{figure}

\section{Proof of dimerization}
\label{sec dimerisation}

\subsection{Setting of the cluster expansion}

We summarize the main results of the method of cluster expansion as we need it. The following setting and theorem was proposed in \cite{ueltschi04}, extending the results of \cite{KP} to the continuous setting and general repulsive interactions.

Let $\Gamma$ be a measurable space, $\eta$ a complex measure on $\Gamma$ such that $|\eta|(\Gamma) < \infty$, where $|\eta|$ is the total variation (absolute value) of $\eta$. Let $\zeta$ be a symmetric function $\Gamma \times \Gamma \to \bbC$ such that $|1+\zeta(\gamma,\gamma')| \leq 1$ for all $\gamma, \gamma' \in \Gamma$. Define the partition function $Z$ by
\be
Z = \sum_{k\geq0} \frac1{k!} \int\dd\eta(\gamma_1) \dots \int\dd\eta(\gamma_k) \prod_{1\leq i < j \leq k} \big( 1 + \zeta(\gamma_i,\gamma_j) \bigr).
\ee
Finally, define the cluster function
\be
\varphi(\gamma_1,\dots, \gamma_k) = \begin{cases} 1 & \text{if } k=1, \\ \frac1{k!} \sum_G \prod_{\{i,j\} \in G} \zeta(\gamma_i,\gamma_j) & \text{otherwise,} \end{cases}
\ee
where the sum is over {\it connected graphs} of $k$ elements, and the product is over the edges of $G$. Then we have the following expressions and estimates.

\begin{theorem}
\label{thm clexp}
Assume that there exist functions $a,b : \Gamma \to [0,\infty)$ such that for all $\gamma \in \Gamma$, we have the following Koteck\'y-Preiss criterion
\be
\label{KP crit}
\int\dd|\eta|(\gamma') |\zeta(\gamma,\gamma')| \e{a(\gamma') + b(\gamma')} \leq a(\gamma).
\ee
(Also, assume that $\int\dd|\eta|(\gamma) \e{a(\gamma)+b(\gamma)} < \infty$.) Then we have the following.
\begin{itemize}
\item[(a)] The partition function is equal to
\[
Z = \exp\biggl\{ \sum_{k\geq1} \int\dd\eta(\gamma_1) \dots \int\dd\eta(\gamma_k) \; \varphi(\gamma_1,\dots,\gamma_k) \biggr\},
\]
where the combined sum and integral converges absolutely.
\item[(b)] For all $\gamma_1 \in \Gamma$,
\[
1 + \sum_{k\geq2} k \int\dd|\eta|(\gamma_2)| \dots \int\dd|\eta|(\gamma_k) \Bigl( \sum_{i=1}^k |\zeta(\gamma,\gamma_i)| \Bigr) |\varphi(\gamma_1,\dots,\gamma_k)| \e{b(\gamma_1) + \dots + b(\gamma_k)} \leq \e{a(\gamma_1)}.
\]
\item[(c)] For all $\gamma \in \Gamma$,
\[
\sum_{k\geq1} \int\dd|\eta|(\gamma_1)| \dots \int\dd|\eta|(\gamma_k) \Bigl( \sum_{i=1}^k |\zeta(\gamma,\gamma_i)| \Bigr) |\varphi(\gamma_1,\dots,\gamma_k)| \e{b(\gamma_1) + \dots + b(\gamma_k)} \leq a(\gamma).
\]
\end{itemize}
\end{theorem}

This theorem can be found in \cite{ueltschi04}, see Theorems 1 and 3 there, as well as Eqs (18) and (19). Notice that 
the term $b(\gamma)$ is not usually part of the Koteck\'y--Preiss criterion and is not needed for convergence of the cluster expansion. But it gives better estimates, see (b) and (c) above, which are most helpful in proving exponential decay.

\subsection{Cluster expansion for the partition function} 

Let us return to our loop model. We start with the partition function \eqref{eq:Z-def}, namely
\be
Z_{\ell,\beta,n,u}=\e{-(1+u)|\overline E_{\ell,\beta}|}
\int_{\Omega_{\ell,\beta}}
\dd\bar\rho_{u}(\omega)
n^{\mathcal L(\omega)-\#\omega_{\dbar}}
\ee
where 
$\dd\bar\rho_{u}(\omega)=
u^{\#\omega_{\cross}} \dd^{\otimes \#\omega}x$
is given in \eqref{eq:unnormalisedPPP}. 
Since we identify contours $\gamma\in X_{\ell,\beta}$ with the links 
they are made up of, $\dd\bar\rho_{u}(\gamma)$ is also well defined.
We define 
\be
	\tilde w(\g)\deq \e{-(1+u)\frac{1}{2}|\gamma|}n^{h(\gamma)}u^{\#\gamma_{\cross}}, 
\ee
where $h(\gamma)$ is defined in \eqref{eq:h-def}. 
Let $\mathcal L(\G)=\{l:\exists \g\in\G:l\in\g\}$ be the set of loops in a (not necessarily admissible) collection of contours $\G$. Let $Y_{\ell,\beta}\subseteq X_{\ell,\beta}$ the set of contours $\g$ (not necessarily admissible) consisting of two adjacent winding loops not traversing any links and let $\mathcal Y_{\ell,\beta}\deq \{\G\se Y_{\ell,\beta}:\g\cap\g'=\emptyset,\;\forall\g\neq\g'\in \G\}$.
Now let 
\be\label{eq:w-def}
w(\gamma)\deq \sum_{\tilde\g\in g(\g)}\tilde w(\tilde \g)(-\e{-2\beta})^{\#\g\setminus\tilde\g\over2}, 
\ee
where 
\be\label{eq:g}
g(\g)=\{\tilde\g\se \g\mid\exists \G'\in\mathcal Y_{\ell,\beta} : \g=\tilde\g\cup\mathcal L(\G')\}
\ee
is the set of contours $\tilde \g$ such that $\g$ can be obtained by adding pairs of adjacent, winding loops not traversing any links (those that come from having an ``empty good column") to $\tilde\g$ and $\#\g\setminus\tilde\g$ denotes the number of loops that are in $\g$, but not in $\tilde\g$ -- necessarily an even number by the definition of $\mathcal Y_{\ell,\beta}$. 

Note that for $\g\in X_{\ell,\beta}^{\rm{nw}}$ we have $g(\g)=\{\g\}$, so $w(\g)=\tilde w(\g)$.

\begin{pro}\label{pro:Z}
We have, for any $u\in\mathbb{R}$,
\[
Z_{\ell,\beta,n,u}  =  \e{-(1+u)|\overline E_\ell^-|}
\sum_{k\geq 0}\frac{1}{k!} \int_{X^+_\ell} \dd\bar\rho_{1}(\gamma_1)
\dots \int_{X^+_\ell} \dd\bar\rho_{1}(\gamma_k) \Bigl( \prod_{i=1}^k w(\gamma_i) \Bigr) \prod_{1\leq i<j\leq  k}\delta(\gamma_i,\gamma_j).
\]
\end{pro}

\begin{proof}
First note that
 $d\bar\rho_{u}$ factorises, 
i.e.\ for $\omega_1,\omega_2\in\Omega_{\ell,\beta}$ sharing no links, 
we have $d\bar\rho_{u}(\omega_1\cup\omega_2)=
d\bar\rho_{u}(\omega_1)d\bar\rho_{u}(\omega_2)$. 
In particular, for any admissible set
$\{\g^1,\dotsc,\g^k\}\in\mathfrak{A}_{\ell,\beta}$
of contours,
\be
\dd\bar\rho_{u}(\g^1\cup\dotsc\cup\g^k)=\prod_{i=1}^k \dd\bar\rho_{u}(\g^i).
\ee
Now let $\G_0\in\mathfrak{C}^+_{\ell,\beta}$ denote a
fixed set of \emph{compatible positive-type} contours and let 
\be\label{eq:compatibleExtensions}
\mathcal{A}(\G_0)=\{\omega\in\Omega_{\ell,\beta}:
\G(\om)=\S(\G_0)\}
\ee
denote the set of link-configurations $\om$ that induce the set of contours
$\G_0$ without adding any new contours.
By considering the admissible set $\G(\om)$ and its shift
$\G_0=\S^{-1}(\G(\om))=\{\g^1,\dotsc,\g^k\}$, we conclude that
\be
\begin{split}
Z_{\ell,\beta,n,u} = \e{-(1+u)|\overline E_\ell|}\sum_{k\geq 0} & \frac{1}{k!}
\int_{X^+_{\ell,\beta}} \dd\bar\rho_{1}(\gamma^1)u^{\#\gamma^1_{\cross}}\dots
\int_{X^+_{\ell,\beta}} \dd\bar\rho_{1}(\gamma^k)u^{\#\gamma^k_{\cross}} \\
&\prod_{1\leq i<j\leq k}\delta(\gamma^i,\gamma^j)
\int_{\mathcal{A}(\G_0)}
\dd\bar\rho_{u}\big(\omega\setminus \S(\G_0)\big)n^{\mathcal L(\omega)-\#\omega_{\dbar}}.
\label{eq:Z}
\end{split}
\ee
We also used Remark \ref{rk:cross}, which tells us that crosses are
never `shared' between distinct contours or between a contour and a
short loop.
(Note that in the last integral, we have the measure $\bar\rho_{u}$
rather than $\bar\rho_{1}$ as in the other integrals.)

Next, applying Lemma \ref{lem:Hh} to write 
$\mathcal L(\omega)-\#\omega_{\dbar}=\sum_{i=1}^k h(\g^i)$, we obtain
\be\label{eq:ThetaFact}
\int_{\mathcal{A}(\G_0)}
\dd\bar\rho_{u}\big(\omega\setminus \S(\G_0)\big)
n^{\mathcal L(\omega)-\#\omega_{\dbar}} 
= \prod_{i=1}^k n^{h(\gamma^i)}
\int_{\mathcal{A}(\G_0)}
\dd\bar\rho_{u}\big(\omega\setminus \S(\G_0)\big).
\ee
For $\G \in\fC^+_{\ell,\beta}$ denote by $\mathfrak W(\G)$ the set of columns of $\bar E_{\ell,\beta}$ where adding a double bar at any height would not change the set of contours $\S(\G)$. 
Using Lemma \ref{lem:freeEdges} and recalling that $|T_\beta|=2\beta$, we get
\be
\begin{split}
 \int_{\mathcal{A}(\Gamma_0)}\dd\bar\rho_{u}\big(\omega\setminus \S(\G_0)\big) &
= \e{(1+u)|F(\Gamma(\om))|}
\big(\e{-|T_\beta|}(\e{|T_\beta|}-1)\big)^{|\mathfrak W(\G_0)|}\\
&= \e{(1+u)|\overline E_\ell^+|} \Bigl(\prod_{\gamma\in\Gamma_0}
\e{-(1+u)\frac{1}{2}|\gamma|} \Bigr)
(1-\e{-2\beta})^{|\mathfrak W(\G_0)|}.
\label{eq:FXPoisson}
\end{split}
\ee
Denote by $\g_w(\G)$ the unique winding contour in $\G$, if it exists, and $\emptyset$ otherwise. 
For all $\G\in\fC^+_{\ell,\beta}$ we have
\be
\begin{split}
	(1-\e{-2\beta})^{|\mathfrak W(\G)|} 
	&= \sum_{W\subseteq \mathfrak W(\G)}\prod_{w\in W} (-\e{-2\beta})\\
	&= \sum_{\G'\in \mathcal Y_{\ell,\beta}}\prod_{\g'\in\G'}(-\e{-2\beta})\times 
		\big(
			\one\big\{\g'_w\in\mathcal A_{\ell,\beta}\big\} \prod_{\g \in\G\setminus\{\g_w(\G)\}}\delta(\g,\g'_w)
		\big)\\
	&= \int_{\fX_{\ell,\beta}}\dd\bar\rho_1(\G')\bigg[\prod_{\g'\in\G'}(-\e{-2\beta}\one\{\g'\in Y_{\ell,\beta}\})\bigg]\\
	&\qquad\times\one\big\{\g'_w\in\mathcal A_{\ell,\beta}\big\} \prod_{\g \in\G\setminus\{\g_w(\G)\}}\delta(\g,\g'_w)\prod_{\g,\g'\in\G'}\one\{\g\cap\g'=\emptyset\},
\label{eq:freeColsRep}
\end{split}
\ee
where $\g'_w\equiv\g'_w(\G,\G')\deq \g_w(\G)\cup \mathcal L(\G')$. For the last equation we used that $\bar\rho_1$ is just the counting measure on subsets of $Y_{\ell,\beta}$ since these loops do not traverse any links.
Intuitively this amounts to summing over ``admissible extensions" $\G'$ of some given set of contours $\G$ and assigning a different weight to these extensions. But instead of integrating over one set of contours and then another set of contours that are treated differently, we might also integrate over one set of contours and then decide which weight to give to each part of the contour. More rigorously, we combine Eqs \eqref{eq:Z}, \eqref{eq:ThetaFact}, \eqref{eq:FXPoisson} and \eqref{eq:freeColsRep} to get 
\be
\begin{split}
	Z_{\ell,\beta,n,u}\e{(1+u)|\overline E_\ell^-|} 
	&= \int_{\fC^+_{\ell,\beta}}\dd\bar\rho_1(\G)\bigg(\prod_{\g\in\G}\tilde w(\g)\bigg)(1-\e{-2\beta})^{|\mathfrak W(\G)|}\\
	&= \int_{\fC^+_{\ell,\beta}}\dd\bar\rho_1(\G)\bigg(\prod_{\g\in\G}\tilde w(\g)\bigg)\int_{\fX_{\ell,\beta}}\dd\bar\rho_1(\G')\bigg[\prod_{\g'\in\G'}(-\e{-2\beta}\one\{\g'\in Y_{\ell,\beta}\})\bigg]\\
	&\qquad\times\one\big\{\g'_w\in\mathcal A_{\ell,\beta}\big\} \prod_{\g \in\G\setminus\{\g_w(\G)\}}\delta(\g,\g'_w)\prod_{\g,\g'\in\G'}\one\{\g\cap\g'=\emptyset\}\\
	&= \int_{\fC^+_{\ell,\beta}}\dd\bar\rho_1(\G)\prod_{\g\in\G}\sum_{\tilde\g\in g(\g)}\tilde w(\tilde \g)(-\e{-2\beta})^{\#\g\setminus\tilde\g\over2}.
\end{split}
\ee
\end{proof}

In what follows we estimate
integrals over contours $\g$ which intersect a given point or
interval.  Winding and non-winding contours are treated
separately; we will actually see that winding contours play a very
limited role for $\beta$ large.   We write 
$X_{\ell,\beta}^{\mathrm{w}}\se X_{\ell,\beta}$ and
$X_{\ell,\beta}^{\mathrm{nw}}\se X_{\ell,\beta}$ for the sets of
winding and non-winding contours, respectively. 
We also write
 $X_{\ell,\beta}^{\mathrm{w}}(k)\se X_{\ell,\beta}$
for the set of winding contours which traverse
 exactly $k$ links,  and we write 
$X_{\ell,\beta}^{\mathrm{nw}}(\bar v,k)\se X_{\ell,\beta}$
 for the set of non-winding 
contours $\g$, which traverse  $k$ links and which
visit the point  $\bar v\in\overline V_\ell$.  Similarly, 
if $I=[(v,s), (v,t)]\se \overline V_{\ell,\beta}$ is an interval 
we write  $X_{\ell,\beta}^{\mathrm{nw}}(I,k)$
 for the set of non-winding
contours $\g$, which traverse  $k$ links and which
intersect $I$.

\begin{lem}\label{lem:vgammaBnd}  
Fix any $\ell\in\N,\beta>0$ $c>0$ and 
points $(v,s),(v,t)\in\overline V_\ell$ with $s<t$.  
Write $I=[(v,s), (v,t)]$.
Then we have
\begin{align}
&  \int_{X_{\ell,\beta}^{\mathrm{nw}}(I,k)}
\dd\bar\rho_{1}(\gamma) \e{-c|\gamma|}\leq 
8^{k-1}c^{-k}\big(1+c|I|\big),
\label{eq:bd-w}
\\
&  \int_{X_{\ell,\beta}^{\mathrm{w}}(k)}
\dd\bar\rho_{1}(\gamma) \e{-c|\gamma|}\leq 
(2\ell+1)^{2\ell+2}(k+1)^{2\ell}8^{k-2}c^{-k}.
\label{eq:bd-nw}
\end{align}
\end{lem}
\begin{proof}
Let us start with the case of non-winding contours
and the case when  $I=\{\bar v\}$ contains only one point.
Elements $\g\in X_{\ell,\beta}^{\mathrm{nw}}(\bar v,k)$ may  be
encoded using tuples
$(t_1,\dots,t_k,l_1,\dots,l_{k-1})\in 
\R_+^k\times (\{\cross,\dbar\}\times\{\mathtt{L},\mathtt{R}\}
\times\{\mathtt{1},\mathtt{2}\})^{k-1}$,
as follows.
\begin{itemize}[leftmargin=*]
\item Consider a  walker started at $\bar v$ and travelling upwards 
until it first encounters the endpoint of a link; 
store the vertical distance traversed as $t_1$.
\item This link can go to the left, $\mathtt{L}$, or to the 
right $\mathtt{R}$; it can be a double bar $\dbar$ or a cross
$\cross$; and
it can be traversed by loops in $\gamma$ once, $\mathtt{1}$, 
or twice, $\mathtt{2}$.  Store this information as $l_1$.
\item Having crossed the link, our walker follows 
$\gamma$ and records vertical distances until 
\emph{previously unexplored} links as $t_i$ and information about 
those links as $l_i$, as before.
\item If a loop is closed and there are still links  that are
  traversed twice by loops in $\gamma$, but have only been visited
  once by our walker, the walker continues walking from such a link
and recording $t_i$ and $l_i$ as before
(we fix some arbitrary rule for selecting the link and the direction
of travel).
\item This procedure is  iterated until the entire contour has been traversed.
\end{itemize}
See Fig.\ \ref{fig:bijCont} for an illustration.

\begin{figure}[hbt]\center
\includegraphics[width=6cm]{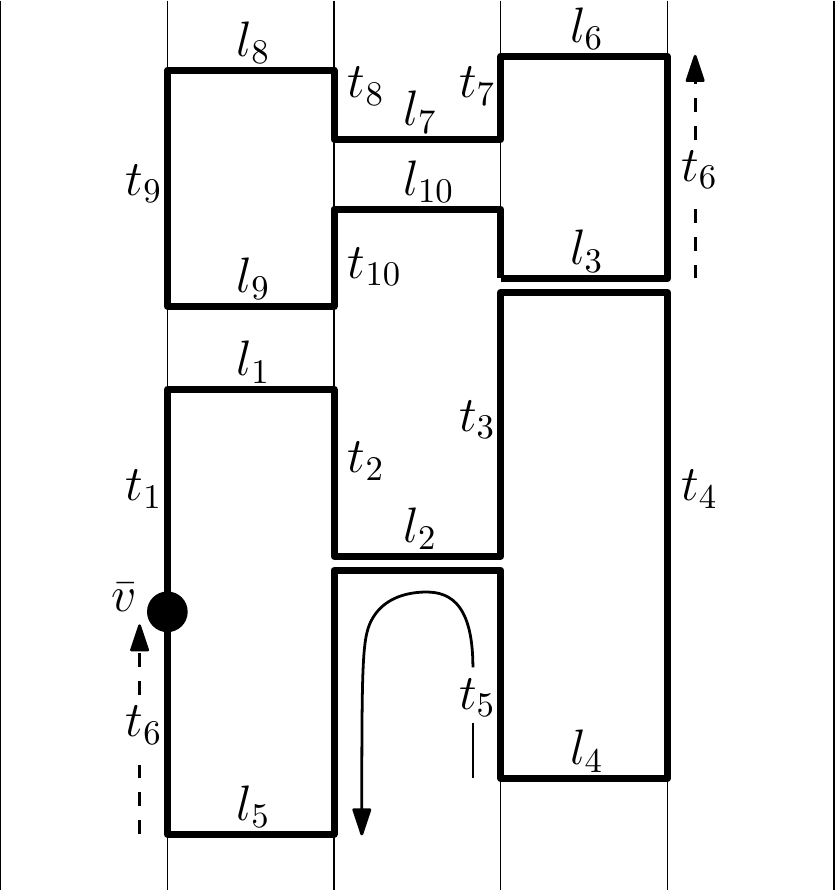}
\caption{Illustration of $t_1,\dots,t_k,l_1,\dots,l_{k-1}$. 
Note, for example,
that $t_5$ is not merely the distance between the fourth
and  fifth links, and that $t_6$ does not readily admit an interpretation
as distance between links at all.}\label{fig:bijCont}
\end{figure}

Noting that $t_1+\dots+t_k\leq |\gamma|$ and that the number of options
for $\{l_i\}_{i=1}^{k-1}$ is bounded by $8^{k-1}$, we get 
\be\label{eq:bd-w1}
\int_{X_{\ell,\beta}^{\mathrm{nw}}(\bar v,k)}
\dd\bar\rho_{1}(\gamma) \e{-c|\gamma|}\leq
8^{k-1}\int_0^\infty dt_1\dots\int_0^\infty \dd t_k\;
\e{-c(t_1+\dots+t_k)}=8^{k-1}c^{-k}.
\ee 
Next, we may apply a similar argument to obtain that, for $\eps>0$
small enough,
\be\label{eq:bd-w2}
\int_{X_{\ell,\beta}^{\mathrm{nw}}((v,t),k)
\setminus X_{\ell,\beta}^{\mathrm{nw}}((v,t+\eps),k)}
\dd\bar\rho_{1}(\gamma) \e{-c|\gamma|}\leq
8^{k-1}c^{-k+1}\tfrac{1-\e{-\eps c}}{c}\leq 
8^{k-1}c^{-k+1}\eps.
\ee
Indeed, for a contour $\g$ which visits $(v,t)$ but not $(v,t+\eps)$,
we must have $t_1\leq\eps$ in the encoding above, and replacing the
integral over $t_1\in[0,\oo)$ with an integral over $t_1\in[0,\eps]$
 gives the claim.  Next, to deduce \eqref{eq:bd-w} from 
\eqref{eq:bd-w1} and \eqref{eq:bd-w2}, we argue as follows.  If $\g$ visits
$I=[(v,s),(v,t)]$, then either $\g$ contains the endpoint $(v,t)$, or 
there are $r\in(s,t)$ and $\eps>0$ such that $(v,r)\in\g$
but $(v,r+\eps)\not\in\g$.  Using \eqref{eq:bd-w1}, the first
possibility accounts for the first term $8^{k-1}c^{-k}$ in
\eqref{eq:bd-w}.  The other possibility accounts for the second term,
which one may, for example, see by using a fine dyadic discretization
of the interval $I$ and passing to the limit using 
\eqref{eq:bd-w2}
and the monotone convergence theorem.

For winding contours $\g$, recall that they consist of $r_1\leq 2\ell+1$ winding loops with a finite number of contractible long loops attached to at least one of them. In particular there are at most $r\leq r_1\leq 2\ell+1$ winding loops that do not share a link and are not connected via a sequence of long, but contractible loops. 
Let us denote these by $\g_1,\dots,\g_r$, and the numbers of links they each visit by $k_1,\dots,k_r$, respectively,
 where $k=\sum_{i=1}^r k_i$.  There are $r$ vertices $v_1,\dots,v_r\in V_\ell$ such
 that  $\g_i$ visits $(v_i,0)$.
Summing over the 
possibilities for $r$,
$v_1,\dots,v_r$, as well as $k_1,\dots,k_r$, 
and applying the argument for \eqref{eq:bd-w1} to each $\g_i$,
we obtain
\be
\begin{split}	
\int_{X_{\ell,\beta}^{\mathrm{w}}(k)}
&\dd\bar\rho_{1}(\gamma) \e{-c|\gamma|} \leq 
\sum_{r=1}^{2\ell+1}{2\ell+1\choose r}\sum_{k_1,\dots,k_r\geq 0\atop k_1+\dots+k_r=k}8^{k-2}c^{-k}\\
&\leq (2\ell+1)(2\ell+1)^{2\ell+1}\max_{1\leq r\leq 2\ell+1} \Bigl| \Bigl\{ (k_1,\dots,k_r)\in\N^r:\sum_i k_i=k \Bigr\} \Bigr| 8^{k-2}c^{-k}\\
&\leq (2\ell+1)^{2\ell+2}(k+1)^{2\ell}8^{k-2}c^{-k}.
\end{split}
\ee
\end{proof}

In order to ensure the convergence of the cluster expansion, we need to check that interactions between contours are small so as to satisfy the Koteck\'y-Preiss criterion in Eq.\ \eqref{KP crit}. For $a_1, a_2, b_1, b_2 \geq 0$, let us introduce 
\be
a(\g)=a_1 |\gamma| + a_2\#\gamma, \qquad b(\g)=b_1 |\gamma| + b_2\#\gamma,
\ee
where $\#\gamma$ denotes the number of links visited by $\g$. Then we have the following bound.

\begin{lem}[Koteck\'y--Preiss criterion]\label{lem:KPcrit}
Let $w(\g)$ be as in \eqref{eq:w-def}. Then
there exist $n_0$, $u_0$, $a_1$, $a_2$, $b_1$, $b_2 >0$ (independent of $\ell, n, u$), and $ \beta_0(\ell,n)$, such that for 
$n>n_0$, $|u|<u_0$, and $\beta > \beta_0(\ell,n)$, we have for any $\ell$ and any $\gamma_0\in X^+_{\ell,\beta}$ that
\be\label{eq:KPcrit}
\int_{X^+_{\ell,\beta}} \dd\bar\rho_{1}(\gamma)
|w(\gamma)|\e{a(\gamma)+b(\gamma)}(1-\delta(\gamma,\gamma_0)) \leq a(\gamma_0).
\ee
\end{lem}


\begin{proof}
Let us alleviate the notation by introducing
\be
\label{def w bar}
\bar w(\gamma)\deq |w(\gamma)|\e{a(\gamma)+b(\gamma)}.
\ee
We use Lemma \ref{lem:nice new lemma} with $\kappa$ such that 
$\frac13- \kappa = \frac14$ and with $u_0$ such that 
$u_0^{-\kappa/2} = n_0$, and  we set
\be
c_1 = \tfrac{1-u_0}2 - a_1 -b_1, \qquad c_2 = \tfrac14 \log n_0 - a_2 - b_2.
\ee
Clearly,
\be\label{eq:contrib from winding and nonwinding}
\int_{X^+_{\ell,\beta}} \dd\bar\rho_{1}(\gamma)
\bar w(\gamma)(1-\delta(\gamma,\gamma_0)) \leq 
\int_{X^\mathrm{w}_{\ell,\beta}} \dd\bar\rho_{1}(\gamma)
\bar w(\gamma) +
\int_{X^\mathrm{nw}_{\ell,\beta}} \dd\bar\rho_{1}(\gamma)
\bar w(\gamma)(1-\delta(\gamma,\gamma_0)).
\ee
Let us first consider the contribution of winding contours.
Notice that
\be
\one\{\gamma\text{ has a spanning segment}\}\leq |\g|/(2\beta).
\ee
For $\g\in X_{\ell,\beta}^{\mathrm{w}}$, note that 
$|w(\g)|\leq|g(\g)|\tilde w_{u=-u_0}(\g)$ 
where $g(\gamma)$ is given in \eqref{eq:g}.
Here $|g(\g)|\leq 2^{|E_\ell|}$, which is some constant depending only
on $\ell$. Hence, also for $\g\in X^{\rm w}_{\ell,\beta}$ we get $\bar
w(\g)\leq \e{-c_1|\g|}\e{-c_2\#\g}$ for $\beta\equiv\beta(\ell,n)$
large enough. 

Using Lemma \ref{lem:vgammaBnd}
and the fact that a winding contour $\g$ satisfies 
$|\g|\geq \tfrac12 |\g|+\beta$, the first term on the right
in \eqref{eq:contrib from winding and nonwinding} satisfies
\be\begin{split}\label{eq:WcontourContribution}
\int_{X^\mathrm{w}_{\ell,\beta}} \dd\bar\rho_{1}(\gamma)
\bar w(\gamma) &\leq
\e{-c_1\beta}\sum_{k\geq0} \e{-c_2 k}
\int_{X^\mathrm{w}_{\ell,\beta}(k)} \dd\bar\rho_{1}(\gamma)
\e{-\tfrac{c_1}{2}|\g|} \\
& \leq (2\ell+1)^{2\ell+2} \e{-c_1\beta}\sum_{k\geq0} 
(k+1)^{2\ell} \big(\tfrac{16}{c_1} \e{-c_2}\big)^k \leq c(\ell,n)\e{-c'\beta},
\end{split}
\ee
with some absolute constant $c'> {1\over4}$ and $c(\ell,n)<\infty$ for $n_0$ sufficiently large such that the geometric series converges. In particular \eqref{eq:WcontourContribution} gets arbitrarily small for $\beta,n_0$ large enough.

We now turn to non-winding contours. We have
\be
\begin{split}
\bar w(\gamma) 
&= \e{-{1+u\over2}|\g|+(a_1+b_1)|\g|}n^{h(\g)}|u|^{\#\g_{\cross}} 
\e{(a_2+b_2)\#\g}  \\
&\leq \e{-(\frac{1-u_0}2 - a_1 - b_1)| \g|} n_0^{-\frac14 \#\g} \e{(a_2+b_2)\#\g} \\
&= \e{-c_1 |\g|}\e{-c_2\#\g}.
\label{eq:wbarBnd}
\end{split}
\ee
Note that if $\delta(\g,\g_0)=0$ then
$\g$ and $\g_0$ intersect somewhere on $\overline V_{\ell,\beta}$.
We may decompose the subset of $\overline V_{\ell,\beta}$ visited by $\g_0$
as a union of closed intervals $I_1,\dotsc,I_m$ where 
$m\leq \#\g_0$ is the number of links of $\g_0$.  Noting also that a
non-winding contour $\g$ has at least 5 links, we obtain 
from Lemma \ref{lem:vgammaBnd} that
\be\begin{split}
\int_{X^\mathrm{nw}_{\ell,\beta}} \dd\bar\rho_1(\gamma)
\bar w(\gamma)(1-\delta(\gamma,\gamma_0)) &\leq
\sum_{j=1}^m\sum_{k\geq 5} \e{-c_2 k}
\int_{X^\mathrm{nw}_{\ell,\beta}(I_j,k)} \dd\bar\rho_1(\gamma) 
\e{-c_1 |\g|}\\
&\leq 
\sum_{j=1}^m\sum_{k\geq 5} \e{-c_2 k}
8^{k-1}c_1^{-k} \big(1+c_1|I_j|\big) \\
&\leq 
\big(\tfrac18 \#\g_0 +\tfrac{c_1}{8}|\g_0|\big)
\sum_{k\geq 5} \big(\tfrac8{c_1} \e{-c_2}\big)^k.
\end{split}\ee
The Lemma \ref{lem:KPcrit} holds true provided that $\beta$ is large enough, and
\be
\tfrac18 \sum_{k\geq 5} \big(\tfrac8{c_1} \e{-c_2}\big)^k \leq a_2, \qquad \tfrac{c_1}8 \sum_{k\geq 5} \big(\tfrac8{c_1} \e{-c_2}\big)^k \leq a_1.
\ee
Both conditions are fulfilled for $n_0$ (and therefore $c_2$) large enough.
\end{proof}

We will also need an estimate on the integral of contours that contain or surround a given point.

\begin{corollary}
\label{cor bound contours}
For any $\eps>0$, there exists $n_0, u_0, a_1, a_2, b_1, b_2 >0$ (independent of $\ell, n, u$) such that for $n>n_0$, $|u|<u_0$, $\beta > \beta_0(\ell,n)$, we have
\be\label{eq:cor bound contours}
\int_{X^+_{\ell,\beta}} \dd\bar\rho_{1}(\gamma)
|w(\gamma)|\e{a(\gamma)+b(\gamma)} \bbone\{ (0,0) \in \overline{I(\gamma)} \} \leq \eps.
\ee
\end{corollary}

\begin{proof}
We proceed as in Eq.\ \eqref{eq:contrib from winding and nonwinding}, so that it suffices to bound the contribution from winding and non-winding contours separately. Recall the definition of $\bar w$ in \eqref{def w bar}.
Using Eq.\ \eqref{eq:WcontourContribution}, we can make the contribution from winding contours arbitrarily small, say $\eps/2$, by choosing $\beta\equiv\beta(\ell,n)$ sufficiently large, i.e.\ we have 
\be
\int_{X^+_{\ell,\beta}} \dd\bar\rho_{1}(\gamma)
\bar w(\gamma) \bbone\{ (0,0) \in \overline{I(\gamma)} \}
\leq \int_{X^\text{nw}_{\ell,\beta}}\dd\bar\rho_1(\g)\bar w(\g)\one\{(0,0)\in\overline{I(\g)}\}+{\eps\over2}.
\ee
If $\g$ is non-winding and has $k$ links, then it must pass by a site at time 0 at distance less than $k/2$ from $(0,0)$. Thus we have the bound
\be
\int_{X^\text{nw}_{\ell,\beta}} \dd\bar\rho_{1}(\gamma)
\bar w(\gamma) \bbone\{ (0,0) \in \overline{I(\gamma)} \} \leq \int_{X^\text{nw}_{\ell,\beta}} \dd\bar\rho_{1}(\gamma)
\bar w(\gamma) \#\gamma\; \bbone\{ (0,0) \in \gamma) \}.
\ee
This can be shown to be arbitrarily small, say less than $\eps/2$, when $n$ and $\beta$ are large, as in the previous lemma.
\end{proof}

Let $\mathcal{C}_k$ denote the set of connected 
(undirected) graphs with vertex set
$\{1,\dotsc,k\}$ and define 
\be
\varphi(\g_1,\dotsc,\g_k)=\left\{
\begin{array}{ll}
1, & \mbox{if } k=1,\\
\tfrac1{k!}\sum_{G\in\mathcal{C}_k} \prod_{ij\in G} 
(\delta(\g_i,\g_j)-1), & \mbox{if } k\geq 2,
\end{array}
\right.
\ee
where the product in the second line is over the edges of $G$.
The following is the main consequence of Theorem \ref{thm clexp}, which holds because of Lemma \ref{lem:KPcrit}.

\begin{pro}[Cluster expansion of the partition function]
\label{pro:cluster}
For parameters as in Lemma \ref{lem:KPcrit}, the following sum
converges absolutely:
\be
\Phi_{\ell,\beta}:=\sum_{m\geq 1} \int_{X^+_{\ell,\beta}}
\dd\bar\rho_1(\gamma_1)
\dots\int_{X^+_{\ell,\beta}} \dd\bar\rho_1(\gamma_k)
\Bigl( \prod_{i=1}^k w(\gamma_i) \Bigr)
\varphi(\g_1,\dotsc,\g_m),
\ee
and we have that
\be
\e{(1+u)|\overline E^-_{\ell,\beta}|} Z_{\ell,\beta,n,u} = \exp\big(\Phi_{\ell,\beta}\big).
\ee
\end{pro}

Notice that $\Phi_{\ell,\beta}$ depends on $n$ and $u$ as well.

\subsection{Dimerization}

Let us introduce the (signed) measure $\mu_{\ell,\beta,n,u}$ such that the integral of a function $f: \Omega_{\ell,\beta} \to \bbC$ is given by
\be\label{eq mu}
\mu_{\ell,\beta,n,u}(f)=\frac1{Z_{\ell,\beta,n,u}} \int_{\Omega_{\ell,\beta}} \dd\rho_{u}(\omega) 
\, n^{\mathcal L(\omega)-\#\omega_{\dbar}} f(\om).
\ee
The next theorem can be understood as dimerization in the loop model. Together with Theorem \ref{thm spin loop} it implies Theorem \ref{thm main}.

\begin{thm}
\label{thm loop dimer}
For any $c>0$, there exist $n_0,u_0>0$ such that for all $ n> n_0$ and $|u|<u_0$, we have
\begin{itemize}
\item[(a)] For all $\ell$ even:
\be
\begin{split}
&\liminf_{\beta\to\infty} \mu_{\ell,\beta,n,u}(0 \overset{-}{\longleftrightarrow} -1) > 1-c, \quad \text{and} \quad \limsup_{\beta\to\infty} |\mu_{\ell,\beta,n,u}(0 \overset{+}{\longleftrightarrow} -1)| < c; \\
&\limsup_{\beta\to\infty} |\mu_{\ell,\beta,n,u}(0 \overset{+}{\longleftrightarrow} 1)| < c, \quad \text{and} \quad \limsup_{\beta\to\infty} |\mu_{\ell,\beta,n,u}(0 \overset{-}{\longleftrightarrow} 1)| < c.
\end{split}
\ee
\item[(b)] For all $\ell$ odd:
\be
\begin{split}
&\liminf_{\beta\to\infty} \mu_{\ell,\beta,n,u}(0 \overset{-}{\longleftrightarrow} 1) > 1-c, \quad \text{and} \quad \limsup_{\beta\to\infty} |\mu_{\ell,\beta,n,u}(0 \overset{+}{\longleftrightarrow} 1)| < c; \\
&\limsup_{\beta\to\infty} |\mu_{\ell,\beta,n,u}(0 \overset{+}{\longleftrightarrow} -1)| < c, \quad \text{and} \quad \limsup_{\beta\to\infty} |\mu_{\ell,\beta,n,u}(0 \overset{-}{\longleftrightarrow} -1)| < c.
\end{split}
\ee
\end{itemize}
\end{thm}

Notice that the limits $\beta\to\infty$ actually exist; this could be established using the correspondence with the quantum spin system, where convergence is clear. This is less visible in the loop model, though, hence the use of $\limsup$ and $\liminf$ so we do not need to prove it.

\begin{proof}
Assume without loss of generality that $\ell$ is odd. The case of $\ell$ even works similarly. Let $\cO$ (for \emph{outside})
denote the event that $(0,0)$ is not on or
inside any contour, that is
\be
\cO =\Big\{\om\in\Omega_{\ell,\beta}: (0,0)\in\bigcap _{\g\in\G(\om)} E(\g)\Big\}.
\ee
By Remark \ref{rk:pos} and the fact that $\mu_{\ell,\beta,n,u}(1)=1$,
\be
\mu_{\ell,\beta,n,u}(0 \overset{-}{\longleftrightarrow} 1) = 1 - \mu_{\ell,\beta,n,u}(\bbone\{ 0 \overset{-}{\longleftrightarrow} 1 \}^{\rm c}) = 1 - \mu_{\ell,\beta,n,u}(\one_{\mathcal O^c} \bbone\{ 0 \overset{-}{\longleftrightarrow} 1 \}^{\rm c}).
\ee
We also have
\be
\begin{split}
&\mu_{\ell,\beta,n,u}(0 \overset{+}{\longleftrightarrow} 1) = \mu_{\ell,\beta,n,u}(\one_{\mathcal O^c} \bbone\{ 0 \overset{+}{\longleftrightarrow} 1 \}^{\rm c}), \\
&\mu_{\ell,\beta,n,u}(0 \overset{\pm}{\longleftrightarrow} -1) = \mu_{\ell,\beta,n,u}(\one_{\mathcal O^c} \bbone\{ 0 \overset{\pm}{\longleftrightarrow} -1 \}^{\rm c}).
\end{split}
\ee
Then Theorem \ref{thm loop dimer} follows from the next lemma (Lemma \ref{lem useful bound}), with the function $f$ being an indicator function.
\end{proof}
Given a set of compatible contours $\G\in \fC^+_{\ell,\beta}$ and its ``shifted version" $\S(\G)\in\fA_{\ell,\beta}$, we identify their contours in the natural way; i.e. for every $\g\in\G$ there exists a unique $\S(\g;\G)\in\S(\G)$ that is obtained by shifting $\g$.
Every $\G\in\fC^+_{\ell,\beta}$ can now be uniquely decomposed into $\G=\G_0\dot\cup \G\setminus \G_0$ with
\be\label{def:Gamma0}
	\G_0\deq \{\g\in\G\mid (0,0)\in\overline{I(\S(\g;\G))}\},
\ee
where $I(\gamma)$ is the interior of $\gamma$, defined in Section \ref{sec contours}.
This decomposition will be useful in the proof of the following lemma.

\begin{lem}
\label{lem useful bound}
Let $g:\Omega_{\ell,\beta}\to\R$ be a function that, for every $\om$, only depends on the contours $\g\in\G(\om)\in\fA_{\ell,\beta}$ that surround or contain $(0,0)$. Assume that $|g|\leq 1$. Then for every $\eps>0$, there exists $n_0\in\N, u_0>0$ such that for all $\ell$, $n>n_0$, $|u|<u_0$, and $\beta\equiv\beta(\ell,n)$ large enough,
\be
	|\mu_{\ell,\beta,n,u}(\one_{\mathcal O^c}g)|<\eps.
\ee 
\end{lem}

\begin{proof}
We have
\be
\mu_{\ell,\beta,n,u}(\one_{\mathcal O^c} g) = \frac{Z_{\ell,\beta,n,u}[\mathcal O^c;g]}{Z_{\ell,\beta,n,u}},
\ee
where
\be
\label{eq:ZGamma}
\begin{split}
	Z_{\ell,\beta,n,u}[\mathcal O^c;g] 
	&= \e{-(1+u)|\overline E_{\ell,\beta}^-|}\int_{\fC^+_{\ell,\beta}\setminus\{\emptyset\}}\dd\bar\rho_1(\Gamma_0) g(\G_0)\Big(\prod_{\gamma\in\Gamma_0}w(\gamma)\Big) \one \Bigl\{ (0,0)\in \overline{I(\S(\gamma;\G_0))}\;\forall\gamma\in\Gamma_0 \Bigr\} \\
	&\times \sum_{m\geq 0}{1\over m!}\int_{X_{\ell,\beta}^+}\dd\bar\rho_1(\gamma_1)\dots\int_{X_{\ell,\beta}^+}\dd\bar\rho_1(\gamma_m) \Bigl(\prod_{i=1}^m w_{\Gamma_0}(\gamma_i) \Bigr) \prod_{1\leq i<j\leq m}\delta(\gamma_i,\gamma_j).
\end{split}
\ee
Notice that the contour weights in the last line depend on $\Gamma_0$ and are defined as
\be
w_{\Gamma_0}(\gamma)\deq w(\gamma)\one\{(0,0)\notin\overline{I(\S(\g;\G_0))}\}\prod_{\g_0\in\G_0}\delta(\g,\g_0).
\ee
Intuitively, we first integrate over all contours surrounding $(0,0)$ (after shifting, they are called $\G_0$) and then we integrate out the remaining contours that are compatible with $\G_0$. 

The second line of Eq.\ \eqref{eq:ZGamma} has the structure of a partition function.
Since $|w_{\Gamma_0}(\g_i)|\leq |w(\g_i)|$,
Lemma \ref{lem:KPcrit} holds for the modified weights too, and
therefore also the suitable modification of 
Proposition \ref{pro:cluster}. This is then equal to $\exp\big(\Phi_{\ell,\beta}(\Gamma_0)\big)$ where
\be
\Phi_{\ell,\beta}(\Gamma_0)\deq\sum_{m\geq 1} \int_{X^+_{\ell,\beta}}
\dd\bar\rho_1(\gamma_1)
\dots\int_{X^+_{\ell,\beta}} \dd\bar\rho_1(\gamma_m)
\Bigl( \prod_{i=1}^m w_{\Gamma_0}(\g_i)\Bigr)
\varphi(\g_1,\dotsc,\g_m).
\ee
Notice that the sum in $\Phi_{\ell,\beta}(\Gamma_0)$ converges absolutely. Then
\bm
	\mu(\one_{\cO^c}g)= \int_{\fC^+_{\ell,\beta}\setminus\{\emptyset\}}\dd\bar\rho_1(\Gamma_0)g(\G_0)\Big(\prod_{\gamma\in\Gamma_0}w(\gamma)\Big) \\
\one \bigl\{ (0,0)\in \overline{I(\S(\gamma;\G_0))}\;\forall\g\in\Gamma_0 \bigr\} \exp\big\{\Phi_{\ell,\beta}(\Gamma_0)-\Phi_{\ell,\beta}\big\}
\label{eq:GammaNSimGamma0}
\end{multline}
Let $\delta_0(\g,\G_0)\deq \one\{(0,0)\notin\overline{I(\S(\g;\G_0))}\}\prod_{\g_0\in\G_0}\delta(\g,\g_0)$ be the indicator function for $\g$ being a contour that is compatible with $\G_0$ and should not be part of $\G_0$. Then
\bm
\Phi_{\ell,\beta}-\Phi_{\ell,\beta}(\Gamma_0) 
= \sum_{m\geq 1}\int_{X^+_{\ell,\beta}}\dd\bar\rho_1(\gamma_1)\dots\int_{X^+_{\ell,\beta}}\dd\bar\rho_1(\gamma_m)\Big(\prod_{i=1}^m w(\gamma_i) \Big )\varphi(\gamma_1,\dots,\gamma_m) \\
\one \bigl\{ \exists i\leq m:\delta_0(\gamma_i,\Gamma_0)=0 \bigr\}.
\end{multline}
We bound these ``corrections coming from contours not in $\G_0$" as follows:
\be
\label{eq:incompatgammasGamma0}
\begin{split}
	&|\Phi_{\ell,\beta}-\Phi_{\ell,\beta}(\Gamma_0)| \\
	&\leq \sum_{m\geq 1}\int_{X^+_{\ell,\beta}}\dd\bar\rho_1(\gamma_1)\dots\int_{X^+_{\ell,\beta}}\dd\bar\rho_1(\gamma_m)\Big(\prod_{i=1}^m |w(\gamma_i)|\Big)|\varphi(\gamma_1,\dots,\gamma_m)| \one\{\exists i\leq m:\delta_0(\gamma_i,\Gamma_0)=0\}\\
	&\leq \int_{X^+_{\ell,\beta}}\dd\bar\rho_1(\g_1)|w(\g_1)|(1-\delta_0(\gamma_1,\Gamma_0)) \biggl( 1 +\sum_{m\geq 2}m\int_{X^+_{\ell,\beta}}\dd\bar\rho_1(\g_2)\dots\int_{X^+_{\ell,\beta}}\dd\bar\rho_1(\g_m)|\varphi(\g_1,\dots,\g_m)| \biggr) \\
	&\leq \int_{X^+_{\ell,\beta}}\dd\bar\rho_1(\g_1)|w(\g_1)|\e{a(\g_1)}(1-\delta_0(\gamma_1,\Gamma_0)).
\end{split}
\ee
The last inequality follows from Theorem \ref{thm clexp} (b). It is easy to see from the definition of $\delta_0$ that $1-\delta_0(\gamma_1,\Gamma_0)\leq \sum_{\g_0\in\G_0}(1-\delta(\g_0,\g_1))+ \one\{(0,0)\in\overline{I(\S(\g_1;\G_0))}\}$. Thus
\be\label{eq:a plus eps bnd}
|\Phi_{\ell,\beta}-\Phi_{\ell,\beta}(\Gamma_0)| \leq \sum_{\g_0\in\G_0}a(\g_0)+\int_{X^+_{\ell,\beta}}\dd\bar\rho_1(\g_1)|w(\g_1)|\e{a(\g_1)}\one\{(0,0)\in \overline{I(\g_1)}\}.
\ee
We used Lemma \ref{lem:KPcrit} for the first summand and translation invariance for the second summand. 
Using Corollary \ref{cor bound contours}, we can bound the second summand by $\eps$, arbitrarily (and uniformly in $\ell,n,u$) small, for $n_0,\beta$ large enough. 

Plugging these bounds back into \eqref{eq:GammaNSimGamma0} and using $|g|\leq 1$, we get 
\be
\label{eq:muOc}
\begin{split}
	\big|\mu_{\ell,\beta,n,u}&(\one_{\cO^c}g)\big| \leq \e{\eps}
	\int_{\fC^+_{\ell,\beta}\setminus\{\emptyset\}}\dd\bar\rho_1(\G_0)\prod_{\g_0\in\G_0}\big(|w(\g_0)| \e{a(\g_0)}\big)\one\{(0,0)\in\overline{I(\S(\g_0;\G_0))}\;\forall \g_0\in\G_0\}\\
	&\leq \e{\eps}
	\sum_{m\geq 1}
	\int_{X_{\ell,\beta}}\dd\bar\rho_1(\g_1)\bar w(\g_1) \one\{(0,0)\in \overline{I(\gamma_1)}\}\;\dots\int_{X_{\ell,\beta}}\dd\bar\rho_1(\g_m)\bar w(\g_m) \one\{(0,0)\in \overline{I(\gamma_m)}\}\\
	&\leq 2\e{\eps}\sum_{m\geq1}\eps^m 
	= 2\e{\eps}\bigg({\eps\over 1-\eps}\bigg).
\end{split}
\ee
For the second inequality we use the fact that all admissible $\S(\G_0)\in\fA_{\ell,\beta}$ such that $(0,0)\in\overline{I(\S(\g_0;\G_0))}\;\forall \g_0\in\G_0$ can be written uniquely as $\{\g_1,\dots,\g_m\}\in\fA_{\ell,\beta}$ with $(0,0)\in \overline{I(\g_m)}$ and $\g_i$ surrounding $\g_j$ whenever $i<j$. In particular all $\g_i$ must surround $(0,0)$. 
The last inequality follows from applying Corollary \ref{cor bound contours} for each of the nested integrals, with a factor of $2$ since we integrate over $X_{\ell,\beta}$ instead of $X_{\ell,\beta}^+$. 

Since $\eps$ can be made arbitrarily small by taking $n_0, \beta$ large, we get the lemma.
\end{proof}

\subsection{Proof of exponential decay of correlations}
\label{sec exp decay}

We now turn to exponential decay. Theorem \ref{thm exp decay} is an immediate consequence of the following result about loop correlations.

\begin{thm}
\label{thm exp decay loops}
There exists an $n_0, u_0, C, c_1, c_2 >0$ (independent of $\ell,n,u$) such that for $n>n_0,|u|<u_0$, we have
\be
\left. \begin{array}{ll}
\big|\mu_{\ell,\beta,n,u}\big((x,s)\leftrightarrow (y,t)\big)\big| \\
\big|\mu_{\ell,\beta,n,u}\big((x,s)\xleftrightarrow{+} (y,t)\big)\big| \\
\big|\mu_{\ell,\beta,n,u}\big((x,s)\xleftrightarrow{-} (y,t)\big)\big|
\end{array}
\right\}\leq C \e{-c_1 |x-y| - c_2 |s-t|}.
\ee
for all $\ell \in \bbN$, all $x,y \in \{-\ell+1,\dots,\ell\}$, and all $s,t \in \R$.
\end{thm}

\begin{proof}
All three bounds can be proved in the same way; here we only discuss the first one. We closely follow the proof of Lemma \ref{lem useful bound} and assume $x=s=0$ for notational convenience.

It turns out that the proof for $|x-y| \leq 1$ present uninformative technical difficulties. On the other hand, exponential decay can be easily proved in the equivalent quantum model by expanding the trace in the basis of eigenvectors of the Hamiltonian and by using the existence of a spectral gap (which is proved in the next section). 

So it is enough to consider here $|x-y|>1$.
This allows us to write $\one\{(x,s)\leftrightarrow(y,t)\}(\om)=\one_{\cO^c}(\om)g(\om)$ with a function $g$ such that $|g|\leq 1$ and that only depends on the contours $\g\in\G(\om)\in\fA_{\ell,\beta}$ that surround or contain $(x,s)$. (Recall that we assumed $(x,s)$ to be the origin $(0,0)$ --- otherwise one would simply have to redefine $\cO$ and $\G_0$ to depend on $(x,s)$.)
We proceed as in Lemma \ref{lem useful bound} to get
\be
	\mu_{\ell,\beta,n,u}\big((x,s)\leftrightarrow (y,t)\big) = \frac{Z_{\ell,\beta,n,u}[\mathcal O^c;g]}{Z_{\ell,\beta,n,u}},
\ee
where $Z_{\ell,\beta,n,u}[\mathcal O^c;g]$ is given as in Eq.\ \eqref{eq:ZGamma}. 
Proceeding exactly the same way, we get the analogue of Eq.\ \eqref{eq:muOc}, namely
\be
\begin{split}
	&\e{c_1|x-y|+c_2|s-t|} |\mu_{\ell,\beta,n,u}(\one_{\cO^c}g)|
	\leq \e{c_1|x-y|+c_2|s-t|} \e\eps \int_{\fC^+_{\ell,\beta}\setminus\{\emptyset\}}\dd\bar\rho_1(\G_0)g(\G_0)\prod_{\g_0\in\G_0}\big(|w(\g_0)| \e{a(\g_0)}\big)\\
	&\hspace{8cm} \times\one\{(x,s)\in\overline{I(\S(\g_0;\G_0))}\forall\g_0\in\G_0\}\\
	&\qquad \leq \e\eps \int_{\fC^+_{\ell,\beta}\setminus\{\emptyset\}}\dd\bar\rho_1(\G_0)\prod_{\g_0\in\G_0}\big(|w(\g_0)|\e{a(\g_0)+b(\g_0)}\big) \one\{(x,s)\in\overline{I(\S(\g_0;\G_0))}\forall\g_0\in\G_0\}\\
	&\qquad \leq 2\e\eps \left({\eps\over1-\eps}\right).
\end{split}
\ee
Here we chose $\eps \in (0,1)$ to be a constant, independent of all other parameters $n,\beta,\ell,u$.

For all $\g\in\G_0$ such that $g(\G_0)\neq 0$ (hence $g(\G_0)=1$) we have $|\gamma| \geq 2|s-t|$ and $\#\gamma \geq 2|x-y|$. Recall the function $b$ of Lemma \ref{lem:KPcrit}. Choosing $c_1 = 2b_1, c_2 = 2b_2$, we get the second inequality. Corollary \ref{cor bound contours} then allows us to proceed as in Eq.\ \eqref{eq:muOc}, which gives the last inequality.
\end{proof}

\section{Proof of the spectral gap}
\label{sec gap}

We follow the method of Kennedy and Tasaki \cite{KT} and show that the method of cluster expansion can be used to prove the existence of a positive spectral gap. Indeed, it implies the validity of the following lemma (recall that $Z_{\ell,\beta} = \Tr \e{-2\beta H_\ell}$).

\begin{lemma}
\label{lem clexp gap}
There exists $n_0, u_0, c>0$ (independent of $\ell,\beta,n,u$) and $C_\ell$ (independent of $\beta,n,u$) such that for all $n \geq n_0$ and $|u| \leq u_0$, we have for all $\beta\geq\frac12$ that
\[
\bigl| E_0^{(\ell)} + \tfrac1{2\beta} \log Z_{\ell,\beta} \bigr| \leq C_\ell \e{-\beta c}.
\]
\end{lemma}

\begin{proof}
We check that, for all $1 \leq \beta < \beta'$, we have
\be
\bigl| \tfrac1{2\beta} \Phi_{\ell\beta} - \tfrac1{2\beta'} \Phi_{\ell\beta'} \bigr| \leq C_\ell \e{-\beta c}.
\ee
Since $-\frac1{2\beta} \log Z_{\ell,\beta} = 2\ell(1+u) - \frac1{2\beta} \log \Phi_{\ell,\beta}$, we get the lemma by taking the limit $\beta' \to \infty$.

Let $\frC_\beta^+$ denote the set of clusters in $V_\ell \times T_\beta$, i.e.\ the sequences of contours $\Gamma = (\gamma_1,\dots,\gamma_k)$, $\gamma_i \in X_{\ell,\beta}^+$, such that $\varphi(\Gamma) \neq 0$. For $t \in T_\beta$, let $\bbone_t(\Gamma)$ be the indicator that the cluster $\Gamma$ crosses the line $V_\ell \times \{t\}$. Let $\frL(\Gamma) \in [0,2\beta]$ be the vertical length of the cluster $\Gamma$:
\be
\frL(\Gamma) = \int_{-\beta}^\beta \bbone_t(\Gamma) \, \dd t.
\ee
Then, from the cluster expansion of Proposition \ref{pro:cluster},
\be
\begin{split}
\frac1{2\beta} \Phi_{\ell,\beta} &= \frac1{2\beta} \int_{\frC_\beta^+} \dd\bar\rho_1(\Gamma) w(\Gamma) \varphi(\Gamma) = \frac1{2\beta} \int_{\frC_\beta^+} \dd\bar\rho_1(\Gamma) \int_{-\beta}^\beta \dd t \frac{\bbone_t(\Gamma)}{\frL(\Gamma)} w(\Gamma) \varphi(\Gamma) \\
&= \frac1{2\beta} \int_{-\beta}^\beta \dd t \int_{\frC_\beta^+} \dd\bar\rho_1(\Gamma) \frac{w(\Gamma) \varphi(\Gamma)}{\frL(\Gamma)} \bbone_t(\Gamma) = \int_{\frC_\beta^+} \dd\bar\rho_1(\Gamma) \frac{w(\Gamma) \varphi(\Gamma)}{\frL(\Gamma)} \bbone_0(\Gamma).
\end{split}
\ee
We used Fubini's theorem to exchange the integrals, and time translation invariance in the last step. The last expression is convenient to cancel terms for different $\beta$s; for $\beta < \beta'$, we have
\be
\label{diff free energies}
\begin{split}
\frac1{2\beta} \Phi_{\ell,\beta} &- \frac1{2\beta'} \Phi_{\ell,\beta'} = \int_{\frC_\beta^+} \dd\bar\rho_1(\Gamma) \frac{w(\Gamma) \varphi(\Gamma)}{\frL(\Gamma)} \bbone_0(\Gamma) - \int_{\frC_{\beta'}^+} \dd\bar\rho_1(\Gamma) \frac{w(\Gamma) \varphi(\Gamma)}{\frL(\Gamma)} \bbone_0(\Gamma) \\
&= \int_{\frC_\beta^+} \dd\bar\rho_1(\Gamma) \frac{w(\Gamma) \varphi(\Gamma)}{\frL(\Gamma)} \bbone_0(\Gamma) \bbone_{\frL(\Gamma) = 2\beta} - \int_{\frC_{\beta'}^+} \dd\bar\rho_1(\Gamma) \frac{w(\Gamma) \varphi(\Gamma)}{\frL(\Gamma)} \bbone_0(\Gamma) 1_{\frL(\Gamma) \geq 2\beta}.
\end{split}
\ee
Indeed, the contribution of clusters with $\frL(\Gamma) < 2\beta$ has precisely canceled.

\begin{figure}[htb]\center
\includegraphics[scale=.5]{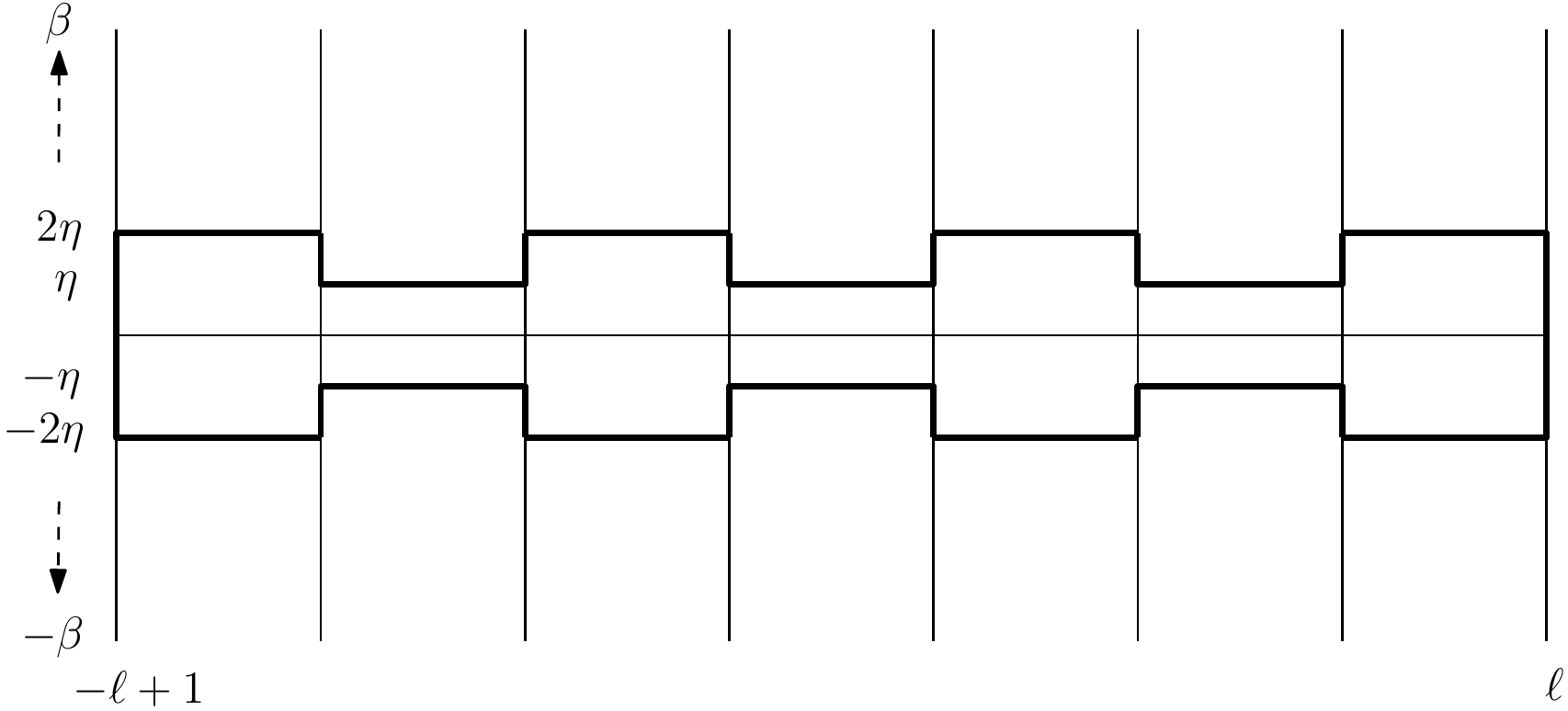}
\caption{The contour $\gamma^{(\eta)}$ used in the proof of Lemma \ref{lem clexp gap}.}
\label{fig contour gap}
\end{figure}

We can use estimates from cluster expansions in order to bound the terms above. For $\eta>0$, let us introduce the contour $\gamma^{(\eta)}$ that surrounds the horizontal axis at time 0, as shown in Fig.\ \ref{fig contour gap}. Its vertical length goes to 0 as $\eta\to0$. We have $a(\gamma^{(\eta)}) = 2(2\ell-1) a_1 + 4(\ell+1) \eta a_2$; and if $\eta<2\beta$, we have
\be
\bbone_0(\Gamma) \leq 1 - \delta(\gamma^{(\eta)}, \Gamma) \leq \sum_{\gamma\in\Gamma} \bigl( 1- \delta(\gamma^{(\eta)},\gamma) \bigr).
\ee
Using $1 \leq 2\beta \leq \sum_{\gamma\in\Gamma} |\gamma|$, which holds for contours such that $1_{\frL(\Gamma) \geq 2\beta}=1$, we have, using Lemma \ref{lem:KPcrit} and the estimate in Theorem \ref{thm clexp} (c) that
\be
\begin{split}
\e{2b_2 \beta} \biggl| \int_{\frC_\beta^+} \dd\bar\rho_1(\Gamma) & \frac{w(\Gamma) \varphi(\Gamma)}{\frL(\Gamma)}  \bbone_0(\Gamma) \bbone_{\frL(\Gamma) = 2\beta} \biggr| \\
&\leq \int_{\frC_\beta^+} \dd\bar\rho_1(\Gamma) |w(\Gamma)| \e{b_2 \sum_{\gamma\in\Gamma} |\gamma|} |\varphi(\Gamma)| \sum_{\gamma\in\Gamma} \bigl( 1- \delta(\gamma^{(\eta)},\gamma) \bigr) \\
&\leq 2(2\ell-1) a_1 + 4(\ell+1) \eta a_2.
\end{split}
\ee
The other term in the right side of \eqref{diff free energies} can be estimated in the same way, giving the same bound. This proves Lemma \ref{lem clexp gap} with $C_\ell = 4(2\ell-1) a_1$ (we can take $\eta\to0$) and $c = 2b_2$.
\end{proof}

\begin{proof}[Proof of Theorem \ref{thm gap}]
With $m_i^{(\ell)}$ the multiplicity of the eigenvalue $E_i^{(\ell)}$ (satisfying $\sum_{i\geq0} m_i^{(\ell)} = n^{2\ell}$), we have
\be
Z_{\ell,\beta} = \sum_{i\geq0} m_i^{(\ell)} \e{-2\beta E_i^{(\ell)}} = m_0^{(\ell)} \e{-2\beta E_0^{(\ell)}} \biggl( 1 + \sum_{i\geq1} \tfrac{m_i^{(\ell)}}{m_0^{(\ell)}} \, \e{-2\beta (E_i^{(\ell)} - E_0^{(\ell)})} \biggr).
\ee
Thus
\be
\label{Z large beta}
-\tfrac1{2\beta} \log Z_{\ell,\beta} = E_0^{(\ell)} - \tfrac1{2\beta} \log m_0^{(\ell)} - \tfrac1{2\beta} R(\ell,\beta),
\ee
where
\be
\label{lower bound for R}
R(\ell,\beta) = \log \biggl( 1 + \sum_{i\geq1} \tfrac{m_i^{(\ell)}}{m_0^{(\ell)}} \, \e{-2\beta (E_i^{(\ell)} - E_0^{(\ell)})} \biggr) \geq \log \biggl( 1 + \tfrac{m_1^{(\ell)}}{m_0^{(\ell)}} \, \e{-2\beta \Delta^{(\ell)}} \biggr) \geq \e{-3\beta \Delta^{(\ell)}},
\ee
for $\beta$ large enough (depending on $\ell$). On the other hand, Lemma \ref{lem clexp gap} implies that
\be
-\tfrac1{2\beta} \log Z_{\ell,\beta} = E_0^{(\ell)} + R'(\ell,\beta),
\ee
where $|R'(\ell,\beta)| \leq C_\ell \e{-\beta c}$. We then have
\be
- \tfrac1{2\beta} \log m_0^{(\ell)} - \tfrac1{2\beta} R(\ell,\beta) = R'(\ell,\beta).
\ee
Using the bound $R(\ell,\beta) \leq n^{2\ell} \e{-2\beta \Delta^{(\ell)}}$ where $\Delta^{(\ell)} > 0$, and looking at the asymptotic $\beta\to\infty$, we see that $m_0^{(\ell)} = 1$. Next, using \eqref{lower bound for R}, we get
\be
\tfrac1{2\beta} \e{-3\beta \Delta^{(\ell)}} \leq C_\ell \e{-\beta c}
\ee
for all $\beta$ sufficiently large; this implies that $\Delta^{(\ell)} \geq \frac13 c$, uniformly in $\ell,n,u$.
\end{proof}

\appendix

\section{The interaction $uT+vP$ when $n$ is even}
\label{app P}

For $n$ odd, the interactions $uT+vP$ and $uT+vQ$ are related by the
unitary transformation of Eq.\ \eqref{takagi_basis}. This holds for
models defined on arbitrary graphs or lattices. 

We now discuss the case of $n$ even. As we shall see, we need to
restrict ourselves to bipartite graphs (of which the chain is of
course an example).  We work with the $S^{(3)}$-eigenbasis
$e_\alpha:=|\alpha\rangle$ with $\alpha = - S, - S + 1, \dots,S$. 
To begin, we define a unitary $V$ by setting 
\be
\label{Bruno's unitary}
V\ket{\alpha} = (-1)^{S-\alpha}\ket{-\alpha}.
\ee
With $\psi$ the vector of \eqref{psi} and $\phi$ the vector of \eqref{phi}, we have
\be
\phi = (\one\otimes V)\psi.
\ee
Therefore, since $P$ is the projection onto $\phi$,
\be
P = (\one\otimes V)Q(\one\otimes V^*).
\ee
Since $T\psi=\psi$ and $T\phi=-\phi$, we have $TQT=Q$ and $TPT=P$. 
Using these properties we find
\be
(V\otimes\one) Q (V^*\otimes \one) = (V\otimes\one) TQT (V^*\otimes \one)
= T(\one\otimes V)Q(\one\otimes V^*)T = P.
\ee
Both models are translation-invariant although the unitary that relates them is not: 
\be
(V\otimes \one\otimes V\cdots \otimes\one) 
\left[ \sum_{x=-\ell+1}^{\ell-1} Q_{x,x+1}\right] 
(V^*\otimes \one\otimes V^*\cdots\otimes \one) = 
\sum_{x=-\ell+1}^{\ell-1} P_{x,x+1}.
\ee
Let $\tilde T$ be the transformation of the operator $T$. We have
\be
\tilde T = (\one \otimes V) T (\one \otimes V^*) = (\one \otimes V)(V^*\otimes \one)  T = - (V\otimes V) T.
\ee

Let us summarize the above considerations by the following proposition. We define the new Hamiltonian $H_\ell' = \sum_{x=-\ell+1}^{\ell-1} \bigl( u \tilde T_{x,x+1} + v Q_{x,x+1} \bigr)$.

\begin{proposition}
\label{prop H'}
For $n$ even, the interaction $uT+vP$ is unitarily equivalent with $u \tilde T + vQ$. The Hamiltonian $\tilde H_\ell$ defined in \eqref{def H tilde} is unitarily equivalent to $H_\ell'$. 
\end{proposition}

Notice that, when $u=0$, the $Q$-model and the $P$-model are unitarily equivalent for all $n$. 
The proposition is stated for chains, but it clearly holds for arbitrary bipartite graphs.

Next, we derive a loop representation for the model $H_\ell'$.

\begin{proposition}
\label{prop H' loops}
There exists a function $s(l)$ from the set of loops to $\pm1$ such that for all $n \geq 2$,
\begin{itemize}
\item[(a)] $\displaystyle \Tr \e{-2\beta H_\ell'} = \e{2\beta (1+u) |E_\ell|} \int \dd\rho_{u}(\omega) n^{\caL(\omega) - |\omega_{\dbar}|} \prod_{\text{loop $l$ in $\omega$}} s(l)$.
\item[(b)] For $i=1,2,3$,
\begin{multline*}
\Tr S_x^{(i)} S_y^{(i)} \e{-2\beta H_\ell'} = \tfrac{n^2-1}{12} \e{2\beta (1+u) |E_\ell|} \int_{\Omega_{\ell,\beta}} \dd\rho_{u}(\omega) \, n^{\mathcal L(\omega)-|\omega_{\dbar}|} \\
\times \bigl( \bbone[x \overset{+}{\longleftrightarrow} y] - \bbone[x\overset{-}{\longleftrightarrow} y] \bigr) \prod_{\text{loop $l$ in $\omega$}} s(l).
\end{multline*}
\end{itemize}
\end{proposition}

Proposition \ref{prop H' loops} is stated for chains, but it actually holds for arbitrary {\it bipartite} graphs (unlike Theorem \ref{thm spin loop} which holds for all finite graphs). For odd $n$ the signs $s(l)$ are all equal to $+1$.

\begin{proof}
First, observe that the number of crosses along the trajectory
of a loop, is even (here, if a cross is traversed twice in a loop, it counts as two). Indeed, the total number of crosses and double-bars
along the trajectory is even because the graph is bipartite; and the
number of double-bars is even because the number of changes in
vertical direction is even; so the number of crosses is also even. 

The expansion of the operator $\e{-2\beta H_\ell'}$ can be made in
terms of configurations $\omega$, and of  ``space-time spin
configurations" (see \cite{Uel}). The space-time spin
configurations that are compatible with $\omega$ have the property
that their value on a loop is $\pm\alpha$ for some 
$\alpha = -S, \dots, S$, the changes of signs occurring when traversing
crosses (and any such choice results in a possible space-time spin
configuration because the number of crosses along a loop trajectory is even). 

Proceeding as in Theorem \ref{thm spin loop}, we find that
\be
\Tr \e{-2\beta H_\ell'} = \e{2\beta (1+u) |E_\ell|} 
\int \dd\rho_{u}(\omega) n^{\caL(\omega) - |\omega_{\dbar}|} \bss(\omega),
\ee
where $\bss(\omega)$ is an overall sign:
$\bss(\omega) = \pm1$.
Notice that, since $\tilde T$ involves a minus sign, there is no need to change the sign of $u$ in the interaction as in Theorem \ref{thm spin loop}.

\begin{figure}[htb]\center
\includegraphics[scale=.7]{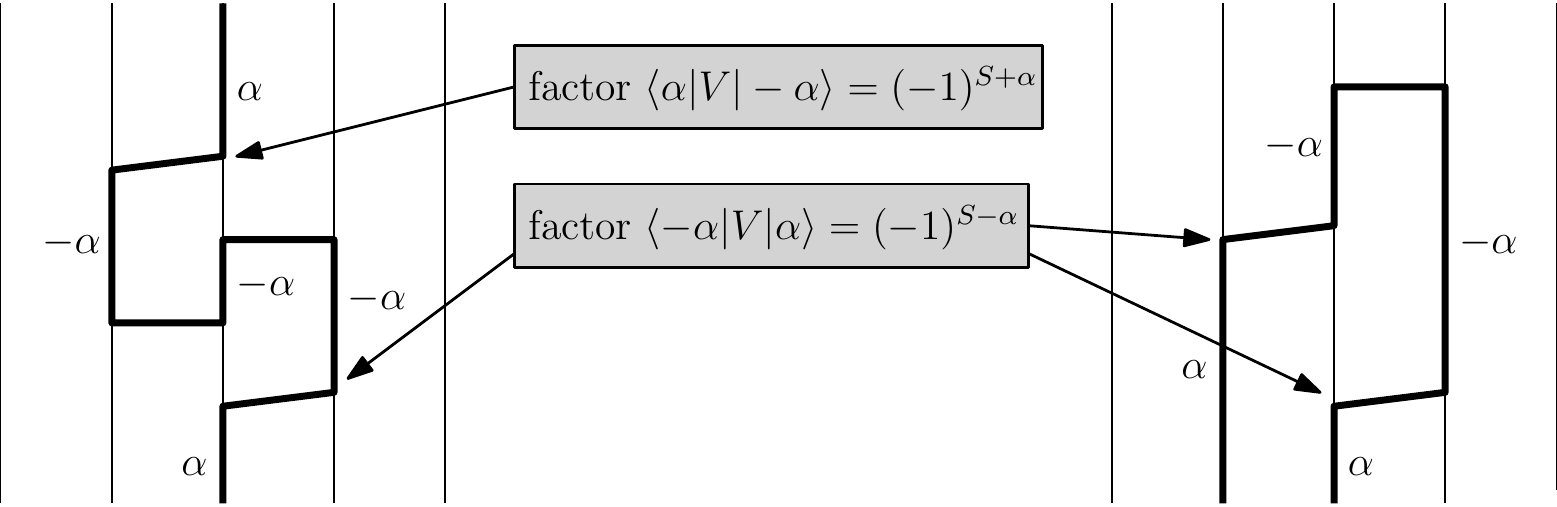}
\caption{Signs arising when traversing crosses. Left: the crosses are separated by an even number of double bars which yields the factor $(-1)^{S-\alpha} (-1)^{S+\alpha} = -1$. Right: the crosses are separated by an odd number of double bars which yields the factor 1.}
\label{fig signs}
\end{figure}

The signs are due to the action of operators $V$. We can collect the
signs for each loop individually. Consider two successive crosses. If
the vertical direction is the same (which is the case if there is an
even number of double-bars between them), we get the factor 
\be
(-1)^{S-\alpha} (-1)^{S+\alpha} = (-1)^{2S} = -1.
\ee
If the vertical direction is opposite (which is the case if there is an odd number of double-bars between them), the factor is
\be
(-1)^{S-\alpha} (-1)^{S-\alpha} = 1.
\ee
This is illustrated in Fig.\ \ref{fig signs}. The value of $s(l)$ is the product of these factors.
Notice that the sign does not depend on the value of $\alpha$ in the
loop. 
This proves item (a) of the proposition. 

The spin correlations are the same for all $i=1,2,3$ by symmetry and it is enough to consider $i=3$. This is identical to \cite[Theorem 3.5 (a)]{Uel} except for the signs (the claim there was restricted to odd $n$ where $s(l)=+1$). Using space-time spin configurations, we have 
\be
\Tr S_x^{(i)} S_y^{(i)} \e{-2\beta H_\ell'} = \e{2\beta (1+u) |E_\ell|} \int_{\Omega_{\ell,\beta}} \dd\rho_{u}(\omega) \, n^{-|\omega_{\dbar}|} \Bigl( \prod_{\text{loop $l$ in $\omega$}} s(l) \Bigr) \sum_{\sigma \in \Sigma(\omega)} \sigma_{x,0} \sigma_{y,0}.
\ee
We used the fact that the signs do not depend on the spin values of the loops. If $(x,0)$ and $(y,0)$ do not belong to the same loop, the sum over $\sigma$ is zero. If $(x,0)$ and $(y,0)$ belong to the same loop and the connection is $x \overset{+}{\longleftrightarrow} y$, then $\sigma_{x,0} = \sigma_{y,0}$ and the sum gives $\frac3 S(S+1)n = \frac1{12}(n^2-1)n$. If the connection is $x \overset{-}{\longleftrightarrow} y$, then $\sigma_{x,0} = -\sigma_{y,0}$ and we get minus the same factor. This gives the identity (b).
\end{proof}

We can now prove Theorem \ref{thm P}.

\begin{proof}[Proof of Theorem \ref{thm P}]
By Proposition \ref{prop H'}, the claims of Theorem \ref{thm P} are
equivalent to proving dimerization in the model with Hamiltonian
$H_\ell'$.  We use the loop representation of Proposition \ref{prop H'
  loops}. We can then retrace the steps of the proof of Theorem
\ref{thm loop dimer}. 
In doing so, note that all short loops $l$ have $s(l)=+1$,
while long or winding loops have $s(l)=\pm 1$.  We incorporate the
latter factors in the weights
$w(\gamma)$ of the contours, see \eqref{eq:w-def}.
Therefore, the only difference is that the weights of 
contours have possibly other signs. All bounds are the same,
though, and the cluster expansion gives the same result.
\end{proof}

The proof of the gap for $\tilde H_\ell$ is exactly the same as the proof for $H_\ell$ described in Section \ref{sec gap}.

\medskip\noindent
{\bf Acknowledgements:}
We are grateful to Vojkan Jak\v si\'c and the Centre de Recherches
Math\'e\-ma\-tiques of Montreal for hosting us during the 
thematic semester ``Mathematical challenges in many-body physics and quantum
information'',  with support from the Simons Foundation through the
Simons--CRM scholar-in-residence program. We also thank the referees for useful comments.

JEB gratefully acknowledges support from \emph{Vetenskapsr{\aa}det} grants
2015-0519 and 2019-04185
as well as \emph{Ruth och Nils-Erik Stenb\"acks stiftelse}.

BN is supported in part by the National Science Foundation under grant DMS-1813149.

\renewcommand{\refname}{\small References}
\bibliographystyle{symposium}

\end{document}